\documentclass[12pt]{amsart}
\usepackage{amsmath,amssymb}
\usepackage{rotating}
\usepackage{xypic}
\usepackage{pst-node,pstricks,multido,pst-plot,pst-text,pst-3d}%
\usepackage{graphicx}
\usepackage{hieroglf}
\def\dfrac#1#2{\frac{\displaystyle #1}{\displaystyle #2}}
\def\binom#1#2{\left(#1\atop#2\right)}
\def\operatorname#1{\mathrm{#1}}
\def\mod{\operatorname{\ mod\ }}
\newtheorem{example}{Example}
\newtheorem{remark}[example]{Remark}
\newtheorem{notation}[example]{Notations}
\newtheorem{theorem}[example]{Theorem}

\newtheorem{corollary}[example]{Corollary}

\newtheorem{definition}[example]{Definition}
\newtheorem{proposition}[example]{Proposition}

\newtheorem{lemma}[example]{Lemma}

\newenvironment{bproof}[1][Proof]{{\noindent\textbf{Proof.}\ }}{\ \\}
\newenvironment{enproof}[1]{\noindent\textbf{#1.} }{\ \rule{0.5em}{0.5em} \\}

\def\S{{\mathfrak  S}}

\def\<{\langle}
\def\>{\rangle}

\def\C{{\mathbb C}}

\def\N{{\mathbb N}}
\def\Y{{\mathbb Y}}

\def\ashuff#1#2#3{
\kern 1pt \vrule height#1 \overline{\vrule height#3 width 0pt
\hskip#2} \rule{.3pt}{#1}\overline{\vrule height#3 width 0pt
\hskip#2} \rule{.3pt}{#1} \kern 1pt }

\def\X{{\mathbb X}}\def\Y{{\mathbb Y}}

\def\Carre3#1{\left[\begin{array}{ccc}#1\end{array}\right]}

\def\X{{\mathbb X}}
\def\Y{{\mathbb Y}}
\def\arm{\mbox{\reflectbox{\Large\pmglyph{A}}}}
\def\leg{\mbox{\pmglyph{b}}}
\def\coarm{\mbox{\Large\pmglyph{A}}}
\def\coleg{\mbox{\reflectbox{\rotatebox[origin=c]{-180}{\pmglyph{b}}}}}
\def\arp{\mathop\rightarrow^+}

\begin{document}
\title[Rectangular Macdonald Polynomials]{Clustering properties of rectangular Macdonald polynomials}
\author{Charles F. Dunkl~ and
Jean-Gabriel Luque ~}
\address{Charles F. Dunkl:\rm ~Dept. of Mathematics, University of
Virginia, Charlottesville VA 22904-4137, US.\\ 
Email:{\tt cfd5z@virginia.edu}.\\
URL: {\tt http://people.virginia.edu/$\sim$cfd5z/}.}
\address{Jean-Gabriel Luque:\rm ~Universit\'e de Rouen, Laboratoire d'Informatique , du Traitement de l'Information et des Syst\`emes (LITIS), Avenue de l’Universit\'e - BP 8
76801 Saint-\'etienne-du-Rouvray Cedex, FR. \\ 
Email:{\tt jean-gabriel.luque@univ-rouen.fr}} 
\begin{abstract}
The clustering properties of Jack polynomials are relevant in the theoretical study of the fractional Hall states. In this context, some factorization properties have been conjectured for the $(q,t)$-deformed problem involving Macdonald polynomials (which are also the quantum eigenfunctions of a familly of commuting difference operators  with signifiance in the relativistic Ruijsenaars-Schneider model). The present paper is devoted to the proof of this formula. To this aim we use four families of Jack/Macdonald polynomials: symmetric homogeneous, nonsymmetric homogeneous, shifted symmetric and shifted nonsymmetric.
\end{abstract}

\maketitle
\section{Introduction}
The symmetric (homogeneous) Jack polynomials are relevant in the study of the quantum many-body wave functions. In particular, fractional quantum Hall states of particles in the lowest Landau levels are described by such polynomials \cite{BH0,BH,BH2,BH3}.
Pioneered by Laughlin \cite{Laugh}, the theoretical study of the fractional quantum Hall states  use multi-variable polynomials \cite{Jain} describing the full many-body state of the interacting electrons on a plane or on a sphere (In the case of the sphere, the polynomials appear after stereographic projection). 
The special polynomials that are relevant in this context are not general solutions of the true eigenvalue problem involving the Coulomb interaction but they are constructed to be adiabatically related to the true eigenstates.  The most famous example is the Laughlin wave function which is the cube of the Vandermonde determinant $\prod_{i<j}(x_i-x_j)^3$ in the variables representing the particles. It is known to be a good approximation of the true state of electrons for the lowest Landau level with filling factor $\frac13$. Another interesting and celebrated example is the Moore-Read Pfaffian \cite{MR,RM}
\[
\Psi_{MR}:={\rm Pf}\left(\frac1{x_k-x_\ell}\right)\prod_{i<j}(x_i-x_j)
\]
which is one of the candidates to approximate the system for the filling factor $\frac52$ and is the polynomial of smallest degree belonging to the kernel of the operator which forbids three particles to be in the same place
\[
\mathcal H=\sum_{i<j<k}\delta^2(x_i-x_j)\delta^2(x_j-x_k).
\]
This operator can be naturally generalized to operators which forbid $k$ particles in the same place, the lowest degree polynomials in its kernel provide other examples of wave-functions called Read-Rezayi \cite{RR1,RR2} states:
\[
\Psi_{RR}^k={\rm Sym}\prod_{\ell=1}^k\prod_{1\leq i_\ell<j_\ell\leq N}(x_{i_\ell}-x_{j_\ell})^2.
\]
This family of wave functions is composed with multivariate symmetric polynomials with some additional vanishing conditions, namely wheel conditions. These polynomials, whose study was pioneered by Feigin, Jimbo, Miwa and Mukhin, are proved to belong to a family of Jack polynomials with negative rational parameter \cite{FJMM1,FJMM2}.
Following the notations of Bernevig and Haldane \cite{BH0,BH,BH2,BH3,BGS}, these polynomials depend upon a parameter $\alpha$ and a configuration of occupation numbers $[n_0,n_1,\dots]$. To recover the standard notation of symmetric functions \cite{Macdo}, the set of occupations defines a decreasing  partition $\lambda=[\lambda_1\geq\lambda_2\dots\geq\lambda_N]$ of the number $N$ of particles. The partition $\lambda$ corresponding to the vector $[n_0,n_1,\dots]$ is such that $n_i$ is the multiplicity of the number $i$ in $\lambda$.
The relevant wave functions belong to the kernel of the differential operators 
\[
L_+:=\sum_i{\partial\over\partial x_i} \mbox{ and } L_-:=N_\phi\sum_{i}x_i-\sum_ix_i^2{\partial\over\partial x_i}\]
where $N_\phi$ denotes the number of flux quanta. 
In the case of the sphere, through the stereographic projection, the natural action of the $SU(2)$ rotations on the quantum states is translated in an action of $L_+$, $L_-$ and $L_z=\frac12NN_\psi-\sum_ix_i{\partial\over\partial x_i}$.
One of the authors (J.-G.L.) with Th. Jolic\oe ur\cite{JL}  described a family of Jack polynomials belonging in the kernel of $L_+$. More precisely, they investigate a $(q,t)$-deformation of the problem involving Macdonald polynomials \cite{Macdo}. \\
Bernevig and Haldane \cite{BH} identified $(k,r)$-clustering properties of some wave functions which relate the functions with $Nk$ variables with those with $N(k-1)$ by a factorization formula of the kind:
\[
\Psi^{(k,r)}(x_1,\dots,x_{(N-1)k},\underbrace{y,\dots,y}_{\times k})=\prod_{\ell=1}^{N-k}(x_\ell-y)^r\Psi^{(k,r)}(x_1,\dots,x_{(N-1)k}).
\]
In particular, this is the case for the Read-Rezayi states:
\[
\Psi^{k}_{RR}(x_1,\dots,x_{(N-1)k},\underbrace{y,\dots,y}_{\times k})=\prod_{\ell=1}^{N-k}(x_\ell-y)^2\Psi^{(k,r)}(x_1,\dots,x_{(N-1)k}).
\]
Furthermore, Bernevig and Haldane \cite{BH0} showed  the connection between the Read-Rezayi states and staircase Jack polynomials $P_{\lambda}^{(\alpha)}$:
\[
\Psi^{k}_{RR}(x)=P^{(-k-1)}_{[(2(N-1))^k,\dots,
2^k,0^k]}(x).
\]
The link between Jack polynomials and quantum Hall states was proven by B. Estienne \emph{et al.}\cite{ES,EBS}. 
Recently, Baratta and Forrester \cite{BF} proved that staircase Jack polynomials for some negative rational parameter satisfy the clustering conditions. More specifically, they stated the result in terms of Macdonald polynomials and recovered the property as a limit case. In the same paper, they conjectured very interesting identities for rectangular partitions (the initial case of staircase with only one step):
\[
P_{r^g}(y,yq^{\frac1\alpha},\dots,yq^{(N-k-1)\over\alpha},x_{N-g},\dots,x_N)=\prod_{\ell=N-g+1}^N\prod_{j=0}^{r-1}(x_\ell-q^{\frac1\alpha}y),
\]  
where $\alpha$ is a certain negative rational number.
The aim of our paper is to prove this conjecture by using some other families of Macdonald polynomials. Indeed, we will use four families of Jack/Macdonald polynomials: symmetric homogeneous, nonsymmetric homogeneous, shifted symmetric, shifted nonsymmetric. The Macdonald polynomials are a two-parameter deformation of the Jack polynomials which can be used to finely understand  equalities involving Jack polynomials. This is due to the fact that they appear in the representation theory of the double affine Hecke algebra (the Jack polynomials are a degenerated version which involves the algebra of the symmetric group). Furthermore they admit a physical interpretation as the eigenfunctions of the Macdonald-Ruijsenaars operator 
\[
H=\sum_{i=1}^N\prod_{j=1\atop i\neq j}^N{e^{x_i}-e^{x_j+\eta}\over e^{x_i}-e^{x_j}}e^{\hbar\partial_{x_i}},
\]
which is a relativistic version of quantum Calogero-Moser system (see \cite{RS1,RS2}). The parameters $q$ and $t$ are related to the parameters $\hbar$ and $\eta$ by $q=e^{\hbar/2}$ and $t=e^{\eta/2}$.

One approach to developing a theory of orthogonal polynomials of several
variables of classical type led to symmetric and nonsymmetric Jack
polynomials. The associated mechanism involves the symmetric group and Young
tableaux. Just as the theory of hypergeometric series is extended to basic
hypergeometric series, the theory of Jack polynomials is extended by the
Macdonald polynomials. Here the symmetric group is replaced by its Hecke
algebra. The orthogonality of the polynomials comes from their realization as
eigenfunctions of a certain set of commuting operators. Generally these
polynomials have coefficients which are rational functions of the parameters
$\left(  q,t\right)  $, assumed not to be roots of unity. Later a further
generalization was developed, namely the theory of shifted Macdonald
polynomials, which are nonhomogeneous polynomials defined by the property of
vanishing on certain points, corresponding to so-called spectral vectors. This
vanishing property leads to expressions of the polynomials as products of
linear factors in various special cases. 

It turns out that for special values of the form $q^{a}t^{b}=1$ for positive
integers $a,b$,  a shifted Macdonald polynomial collapses to its highest
degree term, in which case, it agrees with an ordinary homogeneous Macdonald
polynomial. The label of the polynomial has to satisfy certain restrictions
for this to be possible. These parameter values result in the polynomials
being of \textquotedblleft singular\textquotedblright\ type. In this paper we
concentrate on the rectangular polynomials, meaning those whose leading term
is of the form $\left(  x_{1}x_{2}\ldots x_{k}\right)  ^{m}$. Our results
provide factorizations which are interpreted as clustering properties.

The paper itself starts with an overview of the four different types of
Macdonald polynomials (shifted, ordinary, symmetric, nonsymmetric) and the
algebraic structure of the associated operations. The technical machinery
comprises tableaux, adjacent transpositions, raising operators, and the
Yang-Baxter graph. Then there is a presentation of the binomial formulae of
Knop and Sahi which show how to expand nonhomogeneous Macdonald polynomials in
terms of homogeneous ones. These series are then specialized to the
rectangular versions and this leads to the proofs of our main results. These
all concern the case of $N$ variables with $k\leq N/2$ and the parameters
satisfying $q^{a}t^{N-k+1}=1$ (with $a=m-1$ for the symmetric type and $a=m$
for the nonsymmetric type) but no such relation with smaller exponents. In
closing there is a discussion of future research in the direction of
polynomials of staircase type; this phrase refers to the pictorial
representation of the label of the polynomial, for example $\left(
4,4,2,2,0,0\right)  $.

\section{Macdonald polynomials}
Macdonald polynomials are special functions which are involved in the representation theory of the affine Hecke algebra. This paper is devoted to the case of the symmetric group (type $A_{N+1}$). Up to  normalization Macdonald polynomials are defined to be the simultaneous eigenfunctions of some operators which are a $(q,t)$-deformation of the Cherednik operators. We study four variants of these polynomials: symmetric, nonsymmetric, shifted symmetric and shifted nonsymmetric. Symmetric Macdonald polynomials are indexed by decreasing partitions $\lambda$ whilst the nonsymmetric ones are indexed by vectors $v\in\N^N$. In the aim to simplify the expression arising in the computation, we use the notion of legs and arms of a cell in  the Ferrers diagram of a vector. These numbers are (classically) defined by
\[
\arm_v(i,j):=v[i]-j\]
and
\[ \leg_v(i,j):=\#\{k<i:j\leq v[k]+1\leq v[i]\}+\#\{i<k:j\leq v[k]\leq v[i]\}.
\]
Note if $v=\lambda$ is a partition then $\leg$ is the classical leg-length:
\[
\leg_\lambda(i,j)=\lambda'_i-j
\]
where $\lambda'$ denotes the conjugate of $\lambda$.\\
The leg-length and arm-length are used to define the $(q,t)$-hook product of $v$ with argument $z$:
\[
h_{q,t}(v,z)=\prod_{(i,j)\in v}\left(1-zq^{\arm_v(i,j)}t^{\leg_v(i,j)}\right).
\]
We need also the arm-colength and the leg-colength
\[
\coarm_v(j)=j-1
\]
and

\[ \coleg_v(i):=\#\{k<i:v[i]\leq v[k]\}+\#\{i<k:v[i]< v[k]\}.
\]
Note  the rank function of $v$:
\[
r_v=[1+\coleg_v(1),\dots,1+\coleg_v(N)]
\]
is a permutation of $\S_N$.
\subsection{Affine Hecke Algebra}
We use the notation of Lascoux, Rains and Warnaar\cite{LRW}. Let us recall it here.
Let $N\geq 2$ be an integer and $\X=\{x_1,\dots, x_N\}$ be an alphabet of formal variables. 
We consider the operators $T_i$ acting on Laurent polynomials in the variables $x_j$ by
\begin{equation}\label{defT_i}
T_i=t+(s_i-1){tx_{i+1}-x_i\over x_{i+1}-x_i},
\end{equation}
where $s_i$ is the elementary transposition permuting the variables $x_i$ and $x_{i+1}$.
These operators act on the right  and in particular we have 
\begin{equation}\label{actT_i}
1T_i=t\mbox{ and }x_{i+1}T_i=x_i.
\end{equation}
More precisely, $T_i$ is the unique operator that commutes with symmetric functions in $x_i$ and $x_{i+1}$ satisfying (\ref{actT_i}). The operators $T_i$ satisfy the relations of the Hecke algebra of the symmetric group:
\begin{equation}\label{Hecke}
\begin{array}{l}
T_iT_{i+1}T_i=T_{i+1}T_iT_{i+1}\\
T_iT_j=T_jT_i\mbox{ for }|i-j|>1\\
(T_i-t)(T_i+1)=0
\end{array}
\end{equation}
Together with  multiplication by the variables $x_i$ and the affine operator $\tau$ defined by $f(x)\tau=f(x_N/q,x_1,\dots,x_{N-1})$, they generate the affine Hecke algebra of the symmetric group. More precisely, 
\begin{equation}
{\mathcal H}_N(q,t)=\C(q,t)[x_1^{\pm},\dots,x_ N^{\pm},T_1^{\pm},\dots,T_{N-1}^{\pm},\tau].
\end{equation}

\subsection{Nonsymmetric Macdonald polynomials}
The nonsymmetric Macdonald polynomials $(E_v)_{v\in\N^N}$ are defined as the unique basis of  simultaneous eigenfunctions of the $(q,t)$-version of the Cherednik operators defined by
\begin{equation}
\xi_i:=t^{1-i}T_{i-1}\dots T_{1}\tau T_{N-1}^{-1}\dots T_i^{-1}.
\end{equation}
such that $E_v=x^v+\sum_{u\prec v}\alpha_{u,v}x^u$ with $x^v=x^{v[1]}\dots x^{v[N]}$ if $v=[v[1],\dots,v[N]]$ and $\preceq$ denotes the dominance order on vectors which is based on the dominance order $\preceq_D$ for partitions:
\[
\lambda\preceq_D\mu\mbox{ if and only if } \forall i,\,\lambda_1+\cdots+\lambda_i\leq\mu_1+\cdots+\mu_i.
\]
This order is naturally extended to vectors with the same definition. The dominance order $\preceq$ for vectors is
\[
u\preceq v\mbox{ if and only if either }u^+\preceq_Dv^+\mbox{ or }(u^+=v^+\mbox{ and }u\preceq_Dv)
\]
where $u^+$ denotes the unique nonincreasing partition which is a permutation of $u$.
\\
Note Cherednik operators commute with each other and generate a maximal commutative subalgebra of ${\mathcal H}_N(q,t)$.\\
The corresponding spectral vectors are given by $\mathrm{Spec}_v[i]=\frac1{\langle v\rangle[i]}$ with
\[
\langle v\rangle=[q^{v[1]}t^{N-r_v[1]},\dots,q^{v[N]}t^{N-r_v[N]}].
\]
We recall also the $(q,t)$-Dunkl operators
\[
D_N=(1-\xi_N)x_N^{-1},\, D_i=tT_i^{-1}D_{i+1}T_i^{-1}.
\]
Set 
$
c'_v(q,t)=h_{q,t}(v,q)$.
We will use another normalization which is useful when symmetrizing:
\begin{equation}\label{rmE_v}
\mathrm E_v={t^{n(v)}\over c'_v(q,t)}E_v
\end{equation}
where $n(v)=\sum_{(i,j)\in v}\leg_v(i,j)$.\\
 Knop\cite{Knop} defined and studied the polynomials
\begin{equation}
\mathcal E_v=c''_v(q,t)E_v
\end{equation}
with $c''_v(q,t)= h_{q,t}(v,qt)$.
The expansion of $\mathcal E_v$ on the monomial basis is known to have integral coefficients in $\mathbb{Z}\left[  q,q^{-1},t,t^{-1}\right]  $ \cite{Knop}.
\subsection{Nonsymmetric shifted Macdonald polynomials}
The definition of nonsymmetric shifted Macdonald polynomials $M_v$ is quite similar to those of the Macdonald polynomials $E_v$: This is the  unique basis of  simultaneous eigenfunctions of the Knop-Cherednik operators defined by
\begin{equation}
\Xi_i:=t^{1-i}T_{i-1}\dots T_{1}\tau (1-\frac1{x_N})T_{N-1}^{-1}\dots T_i^{-1}+\frac1{x_i}.
\end{equation}
such that $M_v=q^{-n'(v)}x^v+\sum_{u\prec v}\alpha_{u,v}x^u$ with $x^v=x^{v[1]}\dots x^{v[N]}$ 
if $v=[v[1],\dots,v[N]]$,
\[
n'(v)=\sum_{(i,j)\in v}\arm_v(i,j).
\] 
 and $\preceq$ denotes again the dominance order on 
vectors.\\
Note that, initially, the dominance order is  defined only for vectors with the same norm. We can straightforwardly extend it for any vectors by adding the condition $u\prec v$ when $|u|<|v|$. We remark that the Sahi binomial formula \cite{Sahi}, more precisely  one of its consequences \cite{LRW} (Corollary 4.3), together with a theorem of Knop \cite{Knop2} (see also equation (\ref{0qtbin[]})) imply that if $|u|<|v|$ then $x^u$ has a non null coefficient in the expansion of $M_v$ only if $u^+\subset v^+$. It follows that the natural extension of the dominance order for any pairs of vectors, \emph{i.e.}:
\[
u\preceq v\mbox{ if and only if either } u^+\preceq_D v^+\mbox{ or }(u^+=v^+
\mbox{ and }u\preceq_Dv),
\]
matches with the order appearing in the expansion of $M_v$.
The reader can refer to appendix \ref{Not} for a discussion about the notations.\\
Note operators $\xi_i$ and $\Xi_i$ can be constructed in a similar way, as shown in the following proposition.
\begin{proposition}
For $1\leq i<N$, $\xi_{i}=tT_{i}^{-1}\xi_{i+1}T_{i}^{-1}$ and $\Xi_{i}%
=tT_{i}^{-1}\Xi_{i+1}T_{i}^{-1}$.
\end{proposition}

\begin{proof}
The first part is obvious. For the second part we have%
\[
ptT_{i}^{-1}\Xi_{i+1}T_{i}^{-1}=t^{1-i}pT_{i-1}\ldots T_{1}\tau\left(
1-\frac{1}{x_{N}}\right)  T_{N-1}^{-1}\ldots T_{i}^{-1}+t\left(  \frac
{pT_{i}^{-1}}{x_{i+1}}\right)  T_{i}^{-1},
\]
for any polynomial $p$, and%
\[
\left(  \frac{pT_{i}^{-1}}{x_{i+1}}\right)  T_{i}^{-1}=\frac{p}{tx_{i}}%
\]
thus $tT_{i}^{-1}\Xi_{i+1}T_{i}^{-1}=\Xi_{i}$.
\end{proof}
Also we have:
\begin{proposition}
\label{XixiD}$\Xi_{i}=\xi_{i}+D_{i}$ for $1\leq i\leq N$.
\end{proposition}

\begin{proof}
Start with $i=N$. For any polynomial $p$
\begin{eqnarray*}
p\Xi_{N}  & =t^{1-N}pT_{N-1}\ldots T_{1}\tau\left(  1-\frac{1}{x_{N}}\right)
+\frac{p}{x_{N}}\\
& =p\xi_{N}+\left(  p-p\xi_{N}\right)  /x_{N}\\
& =p\xi_{N}+pD_{N}.
\end{eqnarray*}
The claim follows by downward induction from the relations $D_{i}=tT_{i}%
^{-1}D_{i+1}T_{i}^{-1}$, $\xi_{i}=tT_{i}^{-1}\xi_{i+1}T_{i-1}^{-1}$, and
$\Xi_{i}=tT_{i}^{-1}\Xi_{i+1}T_{i}^{-1}$ .
\end{proof}

Polynomials $M_v$ can be computed by induction from $M_{0^N}$, by the help of the Yang-Baxter graph 
\cite{YangLasc}:
\begin{equation}\label{IntertwinM}
M_{v.s_i}=M_v\left(T_i+{1-t\over 1-{\langle v\rangle[i+1]\over \langle v\rangle[i]}}\right)\mbox{ if }v[i]<v[i+1]
\end{equation}
and
\begin{equation}\label{RaisM}
M_{v\Phi}=M_v\tau(x_N-1),
\end{equation}
where $[v[1],\dots,v[N]]\Phi=[v[2],\dots,v[N],v[1]+1]$ is the raising operator.

Note these polynomials are nonhomogeneous and the spectral vector associated to $M_v$ equals $\mathrm{Spec}_v$. We will consider also the specialization $\mathrm M_v=q^{n'(v)}t^{n(v)}M_v$. 
Alternatively, the Macdonald polynomial $M_v$ can be defined (up a normalization) by interpolation:
\begin{equation}\label{vanish}
M_v(\langle u\rangle)=0\mbox{ for }|v|\leq|u|, u\neq v
\end{equation}
and the coefficient of $x^v$ is $q^{-n'(v)}$.\\
The polynomial $\mathrm E_v$ can be recovered as a limit form from $\mathrm M_v$:
\begin{equation}\label{E2M}
\mathrm E_v(x_1,\dots,x_N)=\lim_{a\rightarrow 0}a^{|v|}\mathrm M_v(x_1/a,\dots,x_N/a).
\end{equation} 

\subsection{Symmetric Macdonald polynomials}
Symmetric Macdonald polynomials are defined as the eigenfunctions of the symmetric polynomials in the variables $\xi_i$. They can be obtained by applying the symmetrizing operator
\begin{equation}\label{S_N}
\mathcal S_N=\sum_{\sigma\in \S_N}T_\sigma
\end{equation}
on the Macdonald polynomial $E_{\lambda^-}$ where $\lambda^-$ is an increasing vector and $T_\sigma=T_{i_1}\dots T_{i_k}$ if $\sigma=s_{i_1}\dots s_{i_k}$ is a shortest decomposition of $\sigma$ in elementary transpositions. The coefficient of $x^\lambda$ in the expansion of $E_{\lambda^-}\mathcal S_N$ is the Poincar\'e polynomial $\phi_t(\S_\lambda)$ where $\S_\lambda$ denotes the stabilizer of $\lambda$ in $\S_N$.\\
We will consider several normalizations of the symmetric Macdonald polynomial. First $P_\lambda$ is such that the coefficient of $x^\lambda$ is $1$. The normalization $\mathrm P_\lambda={t^{n(\lambda)}\over c'_\lambda(q,t)}P_\lambda$ is interesting since the normalization reads:
\begin{equation}\label{P2E}
\mathrm P_\lambda:=\sum_{u^+=\lambda}\mathrm E_u
\end{equation}
where $u^+$ denotes the unique nonincreasing vector which is a permutation of $u$. Finally, we will also use the specialization $J_\lambda=c_{\lambda}(q,t)P_\lambda$ with
$c_v(q,t)=h_{q,t}(v,t)$ which has integral coefficients when expanded in terms of monomials \cite{Knop}.\\
For an infinite alphabet, symmetric Macdonald polynomials can be defined, up to  multiplication by a scalar, as the only  basis having dominance properties which is orthogonal with respect to the scalar product
\begin{equation}\label{scalar}
\langle p_\lambda,p_\mu\rangle_{q,t}=z_\lambda\prod_{i=1}^{\ell(\lambda)}{1-q^{\lambda_i}\over 1-t^{\lambda_i}}\delta_{\lambda,\mu},
\end{equation}
 where $p_\lambda=p_{\lambda_1}\dots p_{\lambda_n}$, $p_k=\sum_{x\in\X}x^k$ is a power sum symmetric function for $k\geq 1$, $\delta_{\lambda,\mu}=1$ if $\lambda=\mu$ and $0$ otherwise and $z_\lambda=\prod_{i=1}^\infty m_i!i^{m_i}$ where $m_i$ denotes the multiplicity of $i$ in $\lambda$.\\
The dual basis of $P_\lambda$ is another normalization of Macdonald polynomials usually denoted by $Q_\lambda$.
\[
\langle P_\lambda,Q_\mu\rangle_{q,t}=\delta_{\lambda,\mu}.
\]
The reproducing kernel
\[
K_{q,t}(\X,\Y)=\sum_{\lambda}P_\lambda(\X)Q_\lambda(\Y)=\sum_\lambda z_\lambda^{-1} \prod_{i=1}^{\ell(\lambda)}{1-t^{\lambda_i}\over 1-q^{\lambda_i}}p_\lambda(\X)p_\lambda(\Y),
\]
admits a nice expression when stated in terms of the Cauchy function $\sigma_1(\X)=\prod_{x\in\X}\frac1{1-x}$ and $\lambda$-ring  \cite{Lasc}:
\begin{equation}\label{ker}
K(\X,\Y)=\sigma_1\left({1-t\over 1-q}\X\Y\right)={\sigma_1\left(\frac1{1-q}\X\Y\right)\over \sigma_1\left(\frac t{1-q}\X\Y\right)},
\end{equation}
where the alphabet $\frac1{1-q}\X\Y$ means $\{q^ixy:i\in\N, x\in\X, y\in\Y\}$ and $\frac t{1-q}\X\Y$ means $\{tq^ixy:i\in\N, x\in\X, y\in\Y\}$.\\
Note formula (\ref{ker}) remains valid for finite alphabets.\\
 Denote $c_\lambda(\X;q,t)=\prod_i{1-t^{\lambda_i}\over 1-q^{\lambda_i}}p_\lambda(\X)$ the dual basis of $p_\lambda$. We have
\[
c_\lambda(\X;q,t)=\prod_i{1-t^{-\lambda_i}\over 1-q^{-\lambda_i}}p_\lambda(\X)=
\left(t\over q\right)^{|\lambda|}\prod_i{t^{-\lambda_i}-1\over q^{-\lambda_i}-1}p_\lambda(\X).
\]
For homogeneous polynomials $R$ and $S$ of global degree $d$ we have
\[
\langle R,S\rangle_{q^{-1},t^{-1}}={q^d\over t^d}\langle R,S\rangle_{q,t}.
\]
Hence,
\begin{equation}\label{Pqm1q}
P_\lambda(\X;q^{-1},t^{-1})={t^{|\lambda|}\over q^{|\lambda|}}P_\lambda(\X;q,t).
\end{equation}
For the normalization $\mathrm P_\lambda$ we have
\[
\mathrm P_\lambda(\X;q^{-1},t^{-1})={t^{-n(\lambda)}\over c_\lambda'(q^{-1},t^{-1})}P_\lambda(\X;q^{-1},t^{-1}).
\]
From 
\[c_\lambda'(q^{-1},t^{-1})=\prod_{(i,j)\in\lambda}\left(1-q^{-1-\arm_\lambda(i,j)}t^{-\leg_\lambda(i,j)}\right)
=
(-1)^{|\lambda|}q^{-|\lambda|-n'(\lambda)}t^{-n(\lambda)}c'_\lambda(q,t)
\]
we obtain
\begin{equation}\label{DPqm1q}
\mathrm P_\lambda(\X;q^{-1},t^{-1})={(-1)^{|\lambda|}t^{|\lambda|}q^{n'(\lambda)}\over c_\lambda'(q,t)} P_\lambda(\X;q,t)=\tau_\lambda \mathrm P_\lambda(\X;q,t).
\end{equation}
with $\tau_v=(-1)^{|v|}q^{n'(v)}t^{-n'(v^+)}$ (using the notation of Lascoux \emph{et al.} \cite{LRW}).
\subsection{Symmetric shifted Macdonald polynomials}
Following the previous subsection, we define the symmetric shifted Macdonald polynomials $MS_\lambda$ as the symmetrization of $M_{\lambda^-}$, where $v^-$ denotes the unique nondecreasing vector which is a permutation of $v$, such that the coefficient of $x^\lambda$ equals $q^{n'(\lambda)}$. More precisely,
\begin{equation}\label{M2MS}
M_{\lambda^-}\mathcal S_N=\phi_t(\S_\lambda) MS_\lambda,
\end{equation}
recall $\mathcal S_N$ is the symmetrizing operator (\ref{S_N}).\\
We will use also the normalization $\mathrm{MS}_\lambda={q^{n'(\lambda)}t^{n(\lambda)}\over c_\lambda'(q,t)} MS_\lambda$. With such a normalization we have
\begin{equation}\label{MS2M}
\mathrm{MS}_\lambda:=\sum_{u^+=\lambda}\mathrm M_u.
\end{equation}
Alternatively, the polynomials $M_\lambda$ are defined by interpolation:
\begin{equation}\label{symvanish}
MS_\lambda(\langle \mu\rangle)=0\mbox{ for }|\lambda|\leq|\mu|, \lambda\neq \mu.
\end{equation}
Note that $q^{-n'(\lambda)}P_\lambda$ is the homogeneous component of maximal degree in $MS_\lambda$, that is
\begin{equation}\label{P2MS}
\lim_{a\rightarrow 0} a^{|\lambda|}MS_\lambda(x_1/a,\dots,x_N/a)=q^{-n'(\lambda)}P_\lambda.
\end{equation}
Equivalently, from equalitions (\ref{E2M}), (\ref{P2E}) and (\ref{MS2M}), we obtain
\begin{equation}\label{DP2DMS}
\mathrm{P}_\lambda(x_1,\dots,x_N)=\lim_{a\rightarrow 0}a^{|\lambda|}\mathrm{MS}_\lambda(x_1/a,\dots,x_N/a).
\end{equation}
That is : $\mathrm{P}_\lambda(x_1,\dots,x_N)$ is the homogeneous component of maximal degree in $\mathrm{MS}_\lambda$.
\section{Generalized binomial coefficients}
\subsection{Definitions}
Sahi \cite{Sahi} generalized binomial coefficients:
\begin{equation}\label{qtbin[]}
\left[u\atop v\right]:={\mathrm M_v(\langle u\rangle)\over \mathrm M_v(\langle v\rangle)}
\end{equation}
From the vanishing properties we have
\begin{equation}
\left[u\atop v\right]=0\mbox{ if }|u|\leq|v|\mbox{ and }u\neq v.
\end{equation}
More generally, a theorem of Knop \cite{Knop2} implies
\begin{equation}\label{0qtbin[]}
\left[u\atop v\right]=0\mbox{ if }v^+\not\subseteq u^+,
\end{equation}
and in particular  $\left[u\atop 0^N\right]=\left[u\atop u\right]=1$.\\ \\
Okounkov \cite{Okoun} and Lassalle \cite{LASS} introduced independently 	the symmetric analogue of this coefficient:
\begin{equation}\label{qtbin()}
\left(u\atop v\right):={\mathrm {MS}_v(\langle u\rangle)\over \mathrm {MS}_v(\langle v\rangle)}
\end{equation}

\subsection{Okounkov binomial formula and consequences}
Symmetric binomial coefficients appear in  Okounkov's binomial formula.\\
Define 
$$\mathrm{MS}'_\lambda(x_1,\dots,x_N)=\frac{q^{|\lambda|}}{\tau_\lambda t^{(N-1)|\lambda|}}\mathrm {MS}_\lambda(t^{1-N}x_1,\dots,t^{1-N}x_N;q^{-1},t^{-1}).$$
\begin{theorem}\label{Okounkov} (Okounkov)
\[
\mathrm{MS}_\lambda(ax_1,\dots,ax_N)=\sum_{\mu}a^{|\mu|}\left(\lambda\atop\mu\right)_{q^{-1},t^{-1}}
{\mathrm{MS}_\lambda(a\langle 0\rangle)\over \mathrm{MS}_\mu(a\langle 0\rangle)}\mathrm {MS}'_\mu(x_1,\dots,x_N)
\]
\end{theorem}

Note from eq. (\ref{DP2DMS}) and (\ref{DPqm1q}), we obtain:
\begin{equation}
\lim_{a\rightarrow 0}a^{|\lambda|}\mathrm{MS'}_\lambda(x_1/a,\dots,x_N/a)=q^{n'(\lambda)}\mathrm P_\lambda.
\end{equation}
Hence, if we apply Theorem \ref{Okounkov} to the alphabet $\{x_1/a,\dots,x_N/a\}$ and we take the limit when $a$ tends to $0$ on the right hand side the we obtain:
\begin{equation}
\mathrm{MS}_\lambda(x_1,\dots,x_N)=\sum_{\mu}q^{n'(\mu)}\left(\lambda\atop\mu\right)_{q^{-1},t^{-1}}
{\mathrm{MS}_\lambda(0,\dots,0)\over \mathrm{MS}_\mu(0,\dots,0)}{\mathrm P}_\mu(x_1,\dots,x_N).
\end{equation}
The constant term $\mathrm{M}_v(0,\dots,0)$ is related to the principal specialization of $\mathrm E_v$ \cite{LRW}:
\begin{equation}\label{M02E<0>}
\mathrm M_v(0,\dots,0)=\tau_v\mathrm E_v(\langle 0\rangle).\end{equation}
Note $\tau_v=\tau_{v^+}$ hence from eq. (\ref{P2E}) and (\ref{MS2M}) we obtain
\begin{equation}\label{MS02E<0>}
\mathrm{MS}_\lambda(0,\dots,0)=\tau_\lambda\mathrm P_\lambda(\langle 0\rangle)=\tau_\lambda{t^{n(\lambda)}\over c'_\lambda(q,t)}P_\lambda(\langle 0\rangle).
\end{equation}
Hence,
\begin{equation}
\mathrm{MS}_\lambda(x_1,\dots,x_N)=q^{n'(\lambda)}\sum_{\mu}
{\tau_\lambda\over\tau_\mu}{t^{n(\lambda)}c'_\mu\over t^{n(\mu)}c'_\lambda}\left(\lambda\atop\mu\right)_{q^{-1},t^{-1}}{P_\lambda(\langle 0\rangle)\over P_\mu(\langle 0\rangle)}{\mathrm P}_\mu(x_1,\dots,x_N).
\end{equation}
Or equivalently:
\begin{lemma}\label{MS2P}
\begin{equation}
MS_\lambda(x_1,\dots,x_N)=\sum_{\mu}
{\tau_\lambda\over\tau_\mu}\left(\lambda\atop\mu\right)_{q^{-1},t^{-1}}{P_\lambda(\langle 0\rangle)\over P_\mu(\langle 0\rangle)} P_\mu(x_1,\dots,x_N).
\end{equation}
\end{lemma}



\section{Principal specializations}

\subsection{Principal specialization of $E_v$}
From equation (\ref{M02E<0>}) the specialization $E_v(\langle 0\rangle)$ is obtained (up to a multiplicative constant $(-1)^*q^\star t^\circ$) from the constant term $M_v(0,\dots,0)$. For simplicity, we sometimes neglect the multiplication by $(-1),\ q$ and $t$ of our polynomials and denote it by $(\star)$.\\
We introduce the classical $q$-Pochhammer symbol: $(a,q)_N=\prod_{i=1}^N(1-aq^{i-1})$ and its generalization for a partition $\lambda$:
\[
(a;q,t)_\lambda=\prod_{i=1}^N(at^{1-i},q)_{\lambda[i]}.
\]
Since $1T_i=t$, we obtain
\begin{equation}\label{IntertwinM0}
M_{v.s_i}(0)=M_v(0)\left(1-t{\langle v\rangle[i+1]\over\langle v\rangle[i]}\over 
1-{\langle v\rangle[i+1]\over\langle v\rangle[i]}\right), if v[i]<v[i+1] 
\end{equation}
from (\ref{IntertwinM}).
From (\ref{RaisM}) we have also
\begin{equation}\label{RaisM0}
M_{v\Phi}(0)=-M_v(0).
\end{equation}
We have:
\begin{proposition}\label{hookM}
For any vector $v\in\N^N$:
\[
M_v(0)=(-1)^{|v|}{(t^Nq;q,t)_{v^+}\over h_{q,t}(v,qt)}
\]
\end{proposition}
\begin{bproof} It suffices to prove that the right hand side $P(v)={(t^Nq;q,t)_{v^+}\over h_{q,t}(v,qt)}$ satisfies the recurrence relations (\ref{IntertwinM0}) and (\ref{RaisM0}).\\
Let us first prove that if $v[i]<v[i+1]$
\begin{equation}\label{Psi}
{P(v.s_i)\over P(v)}=\left(1-t{\langle v\rangle[i+1]\over\langle v\rangle[i]}\over 
1-{\langle v\rangle[i+1]\over\langle v\rangle[i]}\right).
\end{equation}
Note only  $h_{q,t}(v,qt)$ changes in $P(v)$: 
\[
{P(v.s_i)\over P(v)}={h_{q,t}(v.s_i,qt)\over h_{q,t}(v,qt)}.
\]
The only factor in $h_{q,t}(v,qt)$ that changes is the factor for the cell $(i+1,v[i]+1)$.
So
\[
{h_{q,t}(v.s_i,qt)\over h_{q,t}(v,qt)}={1-q^{v[i+1]-v[i]}t^{a+1}\over 1-q^{v[i+1]-v[i]}t^{a}},
\]
where
\[
a=\leg_v(i+1,v[i]+1)=\leg_{v.s_i}(i+1,v[i]+1)+1.
\]
Let
\begin{eqnarray*}
A_{1} &  =\left\{  l:l<i,v\left[  l\right]  \geq v\left[  i\right]  \right\}
\cup\left\{  l:l>i+1,v\left[  i\right]  <v\left[  l\right]  \right\}  ,\\
A_{2} &  =\left\{  l:l<i,v\left[  l\right]  \geq v\left[  i+1\right]
\right\}  \cup\left\{  l:l>i+1,v\left[  i+1\right]  <v\left[  l\right]
\right\}  ,
\end{eqnarray*}
then $A_{2}\subset A_{1}$ and%
\begin{eqnarray*}
A_{1}\backslash A_{2} &  =\left\{  l:l<i,v\left[  i\right]  \leq v\left[
l\right]  <v\left[  i+1\right]  \right\}  \\
&  \cup\left\{  l:l>i+1,v\left[  i\right]  <v\left[  l\right]  \leq v\left[
i+1\right]  \right\}  ;
\end{eqnarray*}
thus $a=\#A_{1}\backslash A_{2}+1$. Also $r_{v}\left[  i\right]  =\#A_{1}+2$
and $r_{v}\left[  i+1\right]  =\#A_{2}+1$, and so $a=r_{v}\left[  i\right]
-r_{v}\left[  i+1\right]  $. This finishes the proof of (\ref{Psi}) since
${\left\langle v\right\rangle [i+1]\over\left\langle v\right\rangle[i]}%
=q^{v\left[  i+1\right]  -v\left[  i\right]  }t^{r_{v}\left[  i\right]
-r_{v}\left[  i+1\right]  }$.\\
\end{bproof}
To prove the second part we need the following lemma
\begin{lemma}\label{DecPart}
Let $\lambda$ be a decreasing partition and $m\leq N$ the biggest integer such that $\lambda[m]>0$. Let 
$$\gamma=[\lambda[1],\dots,\lambda[m-1],\lambda[m]-1,0^{N-m}].$$
If $M_\gamma(0)=P(\gamma)$ then $M_\lambda(0)=P(\lambda)$.
\end{lemma}
\begin{proof}
 Let%
\begin{eqnarray*}
\beta &  =\left[  \lambda\left[  m\right]  -1,\lambda\left[  1\right]
,\ldots,\lambda\left[  m-1\right]  ,0,\ldots\right]  \\
\alpha &  =\left[  \lambda\left[  1\right]  ,\ldots,\lambda\left[  m-1\right]
,0,\ldots0,\lambda\left[  m\right]  \right]  .
\end{eqnarray*}
We evaluate $M_{\beta}\left(  0\right)  /M_{\gamma}\left(  0\right)
,M_{\lambda}\left(  0\right)  /M_{\alpha}\left(  0\right)  $ and use
$M_{\alpha}\left(  0\right)  =-M_{\beta}\left(  0\right)  $ to prove
$M_{\lambda}\left(  0\right)  /M_{\gamma}\left(  0\right)  =P\left(
\lambda\right)  /P\left(  \gamma\right)  $. The formula (\ref{IntertwinM0}) is used
repeatedly in the following calculations. First we have
\[
\frac{M_{\lambda}\left(  0\right)  }{M_{\alpha}\left(  0\right)  }%
=\frac{1-q^{\lambda\left[  m\right]  }t^{N-m+1}}{1-q^{\lambda\left[  m\right]
}t}.
\]
Since $\lambda\left[  m-1\right]  \geq\lambda\left[  m\right]  $ it follows
that $\gamma\left[  i\right]  >\gamma\left[  m\right]  $ for $1\leq i<m$. Then%
\[
\frac{M_{\beta}\left(  0\right)  }{M_{\gamma}\left(  0\right)  }%
=\prod\limits_{i=1}^{m-1}\frac{1-q^{\lambda\left[  i\right]  -\lambda\left[
m\right]  +1}t^{m-i}}{1-q^{\lambda\left[  i\right]  -\lambda\left[  m\right]
+1}t^{m-i+1}}.
\]
Next $\dfrac{\left(  qt^{N};q,t\right)  _{\lambda}}{\left(  qt^{N};q,t\right)
_{\gamma}}=\left(  1-q^{\lambda\left[  m\right]  }t^{N+1-m}\right)  $ and we
evaluate $\dfrac{h_{q,t}\left(  \gamma,qt\right)  }{h_{q,t}\left(
\lambda,qt\right)  }$. In this ratio the only cells that have a changed factor
are those in column $\lambda\left[  m\right]  $ and in row $m$. So
$\leg_\lambda\left( i,\lambda\left[  m\right]  \right)  =m-i$ and 
$\leg_\gamma\left(
i,\lambda\left[  m\right]  \right)  =m-i-1$ for $1\leq i<m$. The cells
in $\left\{  \left(  1,\lambda\left[  m\right]  \right)  ,\cdots\left(
m-1,\lambda\left[  m\right]  \right)  \right\}  $ contribute
\[
\prod\limits_{i=1}^{m-1}\frac{1-q^{\lambda\left[  i\right]  -\lambda\left[
m\right]  +1}t^{m-i}}{1-q^{\lambda\left[  i\right]  -\lambda\left[  m\right]
+1}t^{m-i+1}}%
\]
to the ratio and the cells in row $m$ contribute%
\[
\frac{\prod\nolimits_{i=1}^{\lambda\left[  m\right]  -1}\left(  1-tq^{\lambda
\left[  m\right]  -i}\right)  }{\prod\nolimits_{i=1}^{\lambda\left[  m\right]
}\left(  1-tq^{\lambda\left[  m\right]  +1-i}\right)  }=\frac{1}%
{1-tq^{\lambda\left[  m\right]  }}.
\]
Thus
\[
\frac{\left(  qt^{N};q,t\right)  _{\lambda}h_{q,t}\left(  \gamma,qt\right)
}{\left(  qt^{N};q,t\right)  _{\gamma}h_{q,t}\left(  \lambda,qt\right)
}=\frac{1-t^{N-m+1}q^{\lambda\left[  m\right]  }}{1-tq^{\lambda\left[
m\right]  }}\prod\limits_{i=1}^{m-1}\frac{1-q^{\lambda\left[  i\right]
-\lambda\left[  m\right]  +1}t^{m-i}}{1-q^{\lambda\left[  i\right]
-\lambda\left[  m\right]  +1}t^{m-i+1}}.
\]
This agrees with $-M_{\lambda}\left(  0\right)  /M_{\gamma}\left(  0\right)$ and proves the lemma.
\end{proof}
\begin{enproof}{End of the proof of Proposition \ref{hookM}}
Any vector $v$ can be obtained from $0^N$ by adding $1$  to a null part or to a minimal nonzero part. We will denote this operation by $\displaystyle\arp$. For instance,
\[
[0,0,0]\arp [0,1,0] \arp [0,2,0] \arp  [0,3,0] \arp [1,3,0].
\]
Let $v\arp v'$. Suppose $v'^+=[\lambda_1,\dots,\lambda_m,0,\dots,0]$ then $v^+=[\lambda_1,\dots,\lambda_{m-1},\lambda_m-1,0,\dots,0]$. Hence, if $M_{v^+}(0,\dots,0)=P(v^+)$ then by lemma \ref{DecPart} $M_{v'^+}(0)=P(v'^+)$. And using repeatedly eq. (\ref{Psi}), we obtain $M_{v}(0)=P(v)$ implies $M_{v'}(0)=P(v')$.\\
Since $M_{0^N}(0)=P(0^N)$ the result is shown by a straightforward induction.\end{enproof}
As a direct consequence
\begin{corollary}
\[E_v(\langle0\rangle)=(\star){(t^Nq;q,t)_{v^+}\over h_{q,t}(v,qt)}.\]
\end{corollary}
 Lascoux \cite{Dummies} gave an equivalent expression using infinite vectors
\[v^\infty=[v[1],\dots,v[N],v[1]+1,\dots,v[N]+1,\dots]\]
and
\[\langle v\rangle^\infty=[\langle v\rangle[1],\dots,\langle v\rangle[N],q\langle v\rangle[1],\dots,q\langle v\rangle[N],\dots].\]
Note $\langle v\rangle^\infty\neq\langle v^\infty\rangle$.\\
Lascoux showed 
\begin{equation}\label{Lasc}
P(v)=(-1)^{|v|}\prod_{i=1}^N\prod_{j>i\atop v^\infty[i]>v^\infty[j]}{t{\langle v\rangle^\infty[i]\over \langle v\rangle^\infty[j]}-1\over {\langle v\rangle^\infty[i]\over \langle v\rangle^\infty[j]}-1}.
\end{equation}
\subsection{Principal specialization of $P_\lambda$}
From eq. (\ref{MS02E<0>}) and (\ref{M2MS}) one  has
\begin{equation}\label{P<0>2M0}
P_\lambda(\langle 0\rangle)=(\star)MS_\lambda(0)=(\star){\phi_t(\S_N)\over\phi_t(\S_\lambda)}M_{\lambda^-}(0).
\end{equation}
So one has to determine the value of the constant term $M_{\lambda^-}(0)$.\\
 The $t$-multinomial
coefficient is defined as follows: for any partition $\lambda\in\mathbb{N}%
_{0}^{N}$ let $m_{\lambda}\left(  i\right)  =\#\left\{  j:\lambda\left[
j\right]  =i\right\}  $ (defined for $0\leq i\leq\lambda\left[  1\right]  $)
and let $\binom{N}{m_{\lambda}}_{t}=\left(  t,t\right)  _{N}/\prod
\nolimits_{i=0}^{\lambda\left[  1\right]  }\left(  t,t\right)  _{m_{\lambda
}\left(  i\right)  }$ (note $\left(  t,t\right)  _{n}=\prod\nolimits_{i=1}%
^{n}\left(  1-t^{i}\right)  $). In fact, the Poincar\'{e} series for $\lambda
$, denoted $\phi\left(  \S_{\lambda}\right)  $, equals $\prod
\nolimits_{i=0}^{\lambda\left[  1\right]  }\left(  t,t\right)  _{m_{\lambda
}\left(  i\right)}/(1-t)^{N}  $ where $\S_{\lambda}$ is the stabilizer subgroup
of $\lambda$ in $\S_{N}$ (the symmetric group on $N$ letters), and
$\phi\left(  \S_{N}\right)  =\left(  t,t\right)  _{N}/(1-t)^{N}.$

For application of proposition \ref{hookM}, we find a convenient formula for $M_{\lambda^{-}%
}\left(  0\right)  $ for partitions $\lambda$. 

\begin{proposition}
\label{revhook}For a partition $\lambda\in\mathbb{N}_{0}^{N}$%
\[
M_{\lambda^{-}}\left(  0\right)  =\left(  -1\right)  ^{\left\vert
\lambda\right\vert }\frac{\left(  t^{N};q,t\right)  _{\lambda}}{h_{q,t}\left(
\lambda,t\right)  }\frac{\phi\left(  \mathcal{S}_{\lambda}\right)  }%
{\phi\left(  \mathcal{S}_{N}\right)  }.
\]
\end{proposition}
\begin{proof}
From Proposition \ref{hookM} $M_{\lambda^{-}}\left(  0\right)  =\left(
-1\right)  ^{\left\vert \lambda\right\vert }\dfrac{\left(  t^{N}q;q,t\right)
_{\lambda}}{h_{q,t}\left(  \lambda^{-},qt\right)  }$ so we need to evaluate
the ratios $\dfrac{\left(  t^{N}q;q,t\right)  _{\lambda}}{\left(
t^{N};q,t\right)  _{\lambda}}$ and $\dfrac{h_{q,t}\left(  \lambda
^{-},qt\right)  }{h_{q,t}\left(  \lambda,t\right)  }$. It is easy to see that%
\[
\dfrac{\left(  t^{N}q;q,t\right)  _{\lambda}}{\left(  t^{N};q,t\right)
_{\lambda}}=\prod_{i=1}^{N}\frac{1-t^{N+1-i}q^{\lambda\left[
i\right]  }}{1-t^{N+1-i}}.
\]
If the multiplicity of a particular $\lambda\left[  i\right]  $ is $1$ (in
$\lambda$) then the leg-length of cell $\left(  N+1-i,j+1\right)  $ in
$\lambda^{-}$ is the same as that of the cell $\left(  i,j\right)  $ in
$\lambda$ for $1\leq j<\lambda\left[  i\right]  $. To account for
multiplicities let $c$ be the inverse of $r_{\lambda^{-}}$, that is
$r_{\lambda^{-}}\left[  c_{i}\right]  =i$ for $1\leq i\leq N$. Then
$\lambda^{-}\left[  c_{i}\right]  =\lambda\left[  i\right]  $ (for example
suppose $\lambda\left[  1\right]  >\lambda\left[  2\right]  =\lambda\left[
3\right]  >\lambda\left[  4\right]  $, then $c_{1}=N,c_{2}=N-2,c_{3}=N-1$). As
before
\[
\leg_{\lambda^{-}}\left(c_{i},j+1\right)  =\leg_\lambda\left(i,j\right)  ,1\leq
j<\lambda\left[  i\right]  .
\]
The factor in $h_{q,t}\left(  \lambda^{-},qt\right)  $ at the cell $\left(
c_{i},j+1\right)  $ is $\left(  1-t^{\leg_\lambda\left(i,j\right)
+1}q^{\lambda\left[  i\right]  -j}\right)  $ which is the same as the factor
at the cell $\left(  i,j\right)  $ in $h_{q,t}\left(  \lambda,t\right)  $.
Thus $\dfrac{h_{q,t}\left(  \lambda^{-},qt\right)  }{h_{q,t}\left(
\lambda,t\right)  }$ is the product of the factors at the cells $\left(
c_{i},1\right)  $ in $h_{q,t}\left(  \lambda^{-},qt\right)  $ divided by the
product of the factors at the cells $\left(  i,\lambda\left[  i\right]
\right)  $ in $h_{q,t}\left(  \lambda,t\right)  $, for $1\leq i\leq \ell(\lambda)$, since
there are no cells for the zero parts. The factor at $\left(  c_{i},1\right)
$ is$\left(  1-t^{N+1-i}q^{\lambda\left[  i\right]  }\right)  $. Suppose
$\lambda\left[  i\right]  $ has multiplicity $k$ (that is, $\lambda\left[
i-1\right]  >\lambda\left[  i\right]  =\ldots=\lambda\left[  i+k-1\right]
>\lambda\left[  i+k\right]  $) then $\leg_\lambda\left(i+l-1,\lambda\left[
i\right]  \right)  =k-l$ for $1\leq l\leq k$, and the cells $\left(
i,\lambda\left[  i\right]  \right)  ,\ldots,\left(  i+k-1,\lambda\left[
i\right]  \right)  $ contribute $\prod\nolimits_{l=1}^{k}\left(
1-t^{l}\right)  =\left(  t,t\right)  _{k}$. Thus%
\[
\dfrac{h_{q,t}\left(  \lambda^{-},qt\right)  }{h_{q,t}\left(  \lambda
,t\right)  }=\frac{\prod\nolimits_{i=1}^{\ell(\lambda)}\left(  1-t^{N+1-i}q^{\lambda
\left[  i\right]  }\right)  }{\prod\nolimits_{i=1}^{\lambda\left[  1\right]
}\left(  t,t\right)  _{m_{\lambda}\left(  i\right)  }}.
\]
Combining the ratios we find that
\begin{eqnarray*}
\frac{\left(  t^{N};q,t\right)  _{\lambda}}{h_{q,t}\left(  \lambda,t\right)  }
&  =\dfrac{\left(  t^{N}q;q,t\right) _{\lambda}}{h_{q,t}\left(  \lambda
^{-},qt\right)  }\displaystyle\prod\nolimits_{i=1}^{N}\frac{1-t^{N+1-i}}{1-t^{N+1-i}%
q^{\lambda\left[  i\right]  }}\frac{\prod\nolimits_{i=1}^{\ell(\lambda)}\left(
1-t^{N+1-i}q^{\lambda\left[  i\right]  }\right)  }{\prod\nolimits_{i=1}%
^{\lambda\left[  1\right]  }\left(  t,t\right)  _{m_{\lambda}\left(  i\right)
}}\\
&  =\dfrac{\left(  t^{N}q;q,t\right)  _{\lambda}}{h_{q,t}\left(  \lambda
^{-},qt\right)  }\frac{\displaystyle\left(  t,t\right)  _{N}}{\displaystyle\prod\nolimits_{i=1}%
^{\lambda\left[  1\right]  }\left(  t,t\right)  _{m_{\lambda}\left(  i\right)
}\prod\nolimits_{i=\ell(\lambda)+1}^{N}\left(  1-t^{N+1-i}\right)  }\\
&  =\dfrac{\left(  t^{N}q;q,t\right)  _{\lambda}}{h_{q,t}\left(  \lambda
^{-},qt\right)  }\binom{N}{m_{\lambda}}_{t},
\end{eqnarray*}
since $m_{\lambda}\left(  0\right)  =N-\ell(\lambda)$.
\end{proof}
As a direct consequence of Proposition \ref{revhook} and equality 
(\ref{P<0>2M0}), we obtain
\begin{corollary}
\[P_\lambda(\langle 0\rangle)=(\star)\frac{\left(  t^{N};q,t\right)  _{\lambda}}{h_{q,t}\left(
\lambda,t\right)  }.\]
\end{corollary}
\section{Subrectangular Macdonald polynomials}
\subsection{The denominator of sub-rectangular Macdonald polynomials}
The following lemma shows that if $\lambda\subseteq [m^k,0^{N-k}]$ with $2k\leq N$ then $P_\lambda$ does not have a pole at $q^{m-1}t^{N-k+1}=1$ and $q^{m-1\over d}t^{N-k+1\over d}\neq1$ where $d>1$ is a common factor of $m-1$ and $N-k+1$: 
\begin{lemma}\label{DenPrect}
Let $\lambda\subseteq [m^k,0^{N-k}]$ be a partition with $2k\leq N$. Then $(1-q^{m-1}t^{N-k+1})$ is not a factor of any the denominators of coefficients of $x^v$ in $P_\lambda$.
\end{lemma}
\begin{proof}
It suffices to prove that the integral version $J_\lambda$ does not vanish when $1-q^{m-1}t^{N-k+1}=0$ and $1-q^{m-1\over d}t^{N-k+1\over d}\neq 0$ where $d>1$ divides $m-1$ and $N-k+1$.
In other words, we have to show that $c_\lambda(q,t)$ does not have $1-q^{m-1}t^{N-k+1}$ as a factor.
Suppose $\lambda=[m^{k'},\lambda_{k'+1},\dots,\lambda_{k''},0^{N-k''}]$ with $k'\leq k''\leq k$ and $\lambda_{k''}>0$ . Then the only factors with the relevant power of $q$ that occur in $c_\lambda$ are $(1-q^{\lambda_i-j}t^{\lambda'_j-i+1})$ for $\lambda_i=m$ (so $i=1..k'$) and $j=1$ (so $\lambda'_j=k''$). Hence, for such a pair $(i,j)$ one has $\lambda'_j-i+1\leq k\leq N-k$. This proves the result. 
\end{proof}

\subsection{The $(q,t)$-binomial $\left(m^k\atop \lambda\right)$}
%

Lassalle \cite{LASS} gave an explicit formula for the $(q,t)$-binomial $\left(m^k\atop \lambda\right)$ for $\lambda\subseteq m^k$:
\begin{equation}\label{qtbinmk}
\left(m^k\atop \lambda\right)=\prod_{(i,j)\in\lambda}t^{i-k}{\left(t^{i-1}-q^{j-1}t^k\right)
\left(1-q^{m-j+1}t^{i-1}\right)\over \left(1-q^{\arm_\lambda(i,j)}t^{1+\leg_\lambda(i,j)}\right)
 \left(1-q^{1+\arm_\lambda(i,j)}t^{\leg_\lambda(i,j)}\right)}
\end{equation}
Using $(q,t)$-hook products, the formula reads
\begin{equation}\label{qtbinmk2}\binom{m^{k}}{\lambda}=\left(  -1\right)  ^{\left\vert \lambda\right\vert
}t^{3n\left(  \lambda\right)  -\left(  k-1\right)  \left\vert \lambda
\right\vert }q^{m\left\vert \lambda\right\vert -n\left(  \lambda^{\prime
}\right)  }\frac{\left(  t^{k};q,t\right)  _{\lambda}\left(  q^{-m}%
;q,t\right)  _{\lambda}}{h_{q,t}\left(  \lambda,q\right)  h_{q,t}\left(
\lambda,t\right)  }.
\end{equation}

We deduce
\begin{lemma}\label{denqtbin}
Let $\lambda\subsetneq m^k$ be a strict subrectangular partition with $2k\leq N$. The denominator of the reduced fraction $\left(m^k\atop \lambda\right)$ has no factor $1-q^{m-1}t^{N-k+1}$.
\end{lemma}
\begin{proof}
We examine the factor of the denominator of equality (\ref{qtbinmk}):
\[
D_\lambda:=\prod_{(i,j)\in\lambda} \left(1-q^{\arm_\lambda(i,j)}t^{1+\leg_\lambda(i,j)}\right)
 \left(1-q^{1+\arm_\lambda(i,j)}t^{\leg_\lambda(i,j)}\right).
\]
Since $\lambda$ is a partition, there are only terms of the type $1-q^*t^{\lambda'_j-i+1}$ and $1-q^*t^{\lambda'_j-i}$ in $D_\lambda$.
When $\lambda\subsetneq m^k$, we have $\lambda'_j-i+1<\lambda'_j+1\leq k+1$. And so $k\leq N-k$ implies $\lambda'_j-i+1<N-k+1$. This proves the result.
\end{proof}
Note 
\begin{equation}\label{mkmk}
\left(m^k\atop m^k\right)=\left(m^k\atop 0\right)=1.
\end{equation}
\subsection{Principal specialization for subrectangular partitions}
We consider certain poles of $M_{\lambda^{-}}\left(  0\right)  $ for
$\lambda\subseteq\left[  m^{k},0^{N-k}\right]  $ with $2k\leq N$. In particular,
examine the occurrences of $\left(  1-q^{m-1}t^{N-k+1}\right)  $ in $\left(
t^{N};q,t\right)  _{\lambda}$ and $h_{q,t}\left(  \lambda,t\right)  $. Since
the maximum leg-length of any cell in $\lambda\subset\left[  m^{k}%
,0^{N-k}\right]  $ is $k-1$ only factors of the form $\left(  1-t^{a}%
q^{b}\right)  $ with $1\leq a\leq k<N-k+1$ and $0\leq b\leq m-1$ can appear in
$h_{q,t}\left(  \lambda,t\right)  $. Further
\[
\left(  t^{N};q,t\right)  _{\lambda}=\prod\limits_{i=1}^{k}\prod
\limits_{j=1}^{\lambda\left[  i\right]  }\left(  1-t^{N-i+1}q^{j-1}\right)  .
\]
The factor $\left(  1-q^{m-1}t^{N-k+1}\right)  $ appears in this product
exactly when $\lambda\left[  k\right]  =m$, that is, $\lambda=\left[
m^{k},0^{N-k}\right]  $. Thus $M_{\lambda^{-}}\left(  0\right)  $ has a zero
at $q^{m-1}t^{N-k+1}=1$ for $\lambda=\left[  m^{k},0^{N-k}\right]  $ and no
zeros or poles there when $\lambda\subsetneq\left[  m^{k},0^{N-k}\right] $.
We summarize this in the following lemma
\begin{lemma}\label{DenomP0}
Let $\lambda\subseteq[m^k,0^{N-k}]$ with $2k\leq N$ then the reduced fractions 
$M_{\lambda^-}(0)$ and $P_{\lambda}(\langle 0\rangle)$ have $\left(  1-q^{m-1}t^{N-k+1}\right)$ as a factor if and only if $\lambda=[m^k,0^{N-k}]$.\\
Equivalently, if $(1-q^{m-1}t^{N-k+1})=0$ and for any $d>1$ dividing $m-1$ and $N-k+1$, $(1-q^{m-1\over d}t^{N-k+1\over d})\neq 0$ then
\[
{P_{m^k}(\langle 0\rangle)\over P_{\lambda}(\langle 0\rangle)}=0.
\]
\end{lemma}

Note we can write another proof using eq. (\ref{Lasc}) (see appendix \ref{withDum}).

\section{Rectangular Macdonald polynomials}
\subsection{Clustering properties of $MS_{m^k}$\label{secclust}}

In this section, we will use finite alphabets with different sizes. 
\begin{notation}
For simplicity, 
when we use a partition $\lambda\in\mathbb{N}^{N}$, we mean that the underlying alphabet is $\{x_1,\dots,x_N\}$ and the symmetric shifted monic
Macdonald polynomial will be denoted by $MS_{\lambda}^{\left(  N\right)  }=MS_\lambda(x_1,\dots,x_N)$. In the same way, for
$v\in\mathbb{N}^{N}$ the nonsymmetric shifted monic Macdonald polynomial
will be denoted $M_{v}^{\left(  N\right)  }$.\\
 The length of $v\in\mathbb{N}%
^{N}$ is $\ell\left(  v\right)  :=\max\left\{  i:v\left[  i\right]
>0\right\}  $. For $k\geq \ell(v)$ we will denote also $v^{(k)}:=[v_1,\dots,v_k]$.
\end{notation}

The following proposition shows that if $v\in \mathbb N^N$ with $\ell(v)\leq k$ then the Macdonald  polynomials $M_v^{(N)}$ and $M_{v^{(k)}}^{(k)}$ are clearly related.
\begin{proposition}The following assertions hold:
\begin{enumerate} 
\item
Suppose $v\in\mathbb{N}_{0}^{N}$ satisfies $\ell\left(  v\right)  \leq k$ for
some $k<N$, then
\[
M_{v}\left(  x_{1}t^{N-k},\ldots,x_{k}t^{N-k}%
,t^{N-k-1},\ldots,t,1\right)  =t^{\left(  N-k\right)  \left\vert v\right\vert
}M_{v^{(k)}}^{(k)}  .
\]
\item Suppose a partition $\lambda\in\mathbb{N}_{0}^{N}$ satisfies $\ell\left(
\lambda\right)  \leq k$ for some $k<N$, then
\[
MS_{\lambda}^{\left(  N\right)  }\left(  x_{1}t^{N-k},\ldots,x_{k}%
t^{N-k},t^{N-k-1},\ldots,t,1\right)  =t^{\left(  N-k\right)  \left\vert
v\right\vert }MS_{\lambda}^{\left(  k\right)  }\left(  x\right)  .
\]
\end{enumerate}
\end{proposition}

\begin{proof}
We prove first (1):
Both sides have the same coefficient of $x^{v}$ and satisfy the same vanishing
conditions. More precisely, the both sides vanish for $x=\left\langle u\right\rangle $, $u\in\mathbb{N}%
_{0}^{k},\left\vert u\right\vert \leq\left\vert v\right\vert ,u\neq v$, and
$\left\langle u\right\rangle \left[  j\right]  =q^{u\left[  j\right]
}t^{k-r_{u}\left[  i\right]  }$, for $1\leq j\leq k$; and
the multiplicative coefficient is obtained considering the coefficient of the dominant term in the left hand side and the right hand side.\\
To prove (2) remark the left-hand side is symmetric in $\left(  x_{1},\ldots,x_{k}\right)  $. Hence
an obvious modification of the above argument applies.
\end{proof}

\begin{example}\rm
For instance consider $v=[2,0]$. The polynomial $M_{[2,0]}(x_1,x_2)$  vanishes for $[x_1,x_2]=\langle 00\rangle=[t,1]$, $[x_1,x_2]=\langle 10\rangle=[qt,1]$, $[x_1,x_2]=\langle 01\rangle=[1,qt]$, $[x_1,x_2]=\langle 11\rangle=[qt,q]$, $[x_1,x_2]=\langle 02\rangle=[1,q^2t]$ whilst $M_{[20]}(\langle 20\rangle)\neq 0$ (recall that $\langle20\rangle=[q^2t,1]$).\\
Hence, the polynomial $P(x_1)=M_{[2,0]}(x_1t,1)$ vanishes for $[x_1]=[1]=\langle 0\rangle$ and $[x_1]=[q]=\langle 1\rangle$ whilst $P(\langle 2\rangle)=M_{[2,0]}(q^2t,1)\neq 0$. This agrees with the definition of $M_{[2]}(x_1)$ up to a multiplicative factor.
\end{example}

For almost rectangular partitions $v_{m,k}:=[m^{N-k},(m+1)^k]$ we prove that the corresponding Macdonald polynomials nicely factorize:

\begin{proposition}
For $m\geq0$ and $0\leq k<N$  we have%
\[
M_{v_{mk}}^{\left(  N\right)  }=\left(  -1\right)  ^{Nm+k}q^{m\left(
k+\left(  m-1\right)  N/2\right)  }\prod\limits_{i=1}^{N-k}\left(
x_{i},q^{-1}\right)  _{m}\prod\limits_{i=N-k+1}^{N}\left(  x_{i}%
,q^{-1}\right)  _{m+1}.
\]

\end{proposition}

\begin{proof}
Recall $u.\Phi=\left[  u\left[  2\right]  ,\ldots,u\left[  N-1\right]
,u\left[  1\right]  +1\right]  $, thus $v_{m,k}.\Phi=v_{m,k+1}$ for $0\leq
k<N$ and $v_{n,N-1}.\Phi=v_{m+1,0}$. The claim is obvious for $v_{0,0}$.
Suppose the claim is true for some $\left(  m,k\right)  $, then%
\[
M_{v_{m,k}}^{\left(  N\right)  }\tau=q^{m\left(  k+\left(  m-1\right)
N/2\right)  }\left(  -1\right)  ^{Nm+k}\prod\limits_{i=1}^{N-k-1}\left(
x_{i},q^{-1}\right)  _{m}\prod\limits_{i=N-k}^{N-1}\left(  x_{i}%
,q^{-1}\right)_{m+1}\left(  q^{-1}x_{N},q^{-1}\right)_{m}.
\]
Multiply both sides by $\left(  x_{N}-1\right)  $ to get an expression for
$M_{v_{m,k+1}}^{\left(  N\right)  }$ (or $v_{m+1,0}$); observe $\left(
x_{N}-1\right)  \left(  q^{-1}x_{N},q^{-1}\right)_{m}=-\left(  x_{N}%
,q^{-1}\right)_{m+1}$. This completes the inductive proof.
\end{proof}

As a straightforward consequence, setting $k=0$ in the previous proposition, $M_{v_{n,0}}$ is symmetric in $x_1,\dots,x_N$ and then we find:
\begin{corollary}
For $n\geq1$ $MS_{v_{n,0}}^{\left(  k\right)  }=q^{kn\left(  n-1\right)
/2}\left(  -1\right)  ^{kn}\prod_{i=1}^{k}\left(  x_{i},q^{-1}\right)  _{n}$;
note $v_{n,0}=\left[  n^{k}\right] .$
\end{corollary}

Let $\X=\{x_1,\dots,x_N\}$ and $\Y=\{y_1,\dots,y_k\}$ be two alphabets. We will denote by $R(\X,\Y)=\prod_{i=1}^N\prod_{j=1}^k(x_i-y_j)$ the resultant of $\X$ with respect to $\Y$.\\
Let also 
$$\mathcal D_{N,k,m}^{y}(x_1,\dots,x_k):=R([x_1,\dots,x_k],[yt^{N-k}q^0,\dots,yt^{N-k}q^{m-1}]).$$
We have:
\begin{proposition}\label{MS2D}
Let $\lambda=\left[  m^{k},0^{N-k}\right]  $ for some $1\leq k<N$ and $m\geq1$, and $y\neq 0$
then%
\[
MS_{\lambda}^{\left(  N\right)  }\left(  {x_{1}\over y},\ldots,{x_{k}\over y},t^{N-k-1}%
,\ldots,1\right)  =y^{-km}\mathcal D_{N,k,m}^y(x_1,\dots,x_k).
\]
\end{proposition}

\begin{proof}
By the previous propositions%
\[
MS_{\lambda}^{\left(  N\right)  }\left(  {x_{1}\over y},\ldots,{x_{k}\over y},t^{N-k-1}%
,\ldots,1\right)  =t^{k\left(  N-k\right)  m}q^{\frac{m}{2}k\left(
m-1\right)  }\left(  -1\right)  ^{km}\prod_{i=1}^{k}\left(  {x_{i}\over y
t^{N-k}},q^{-1}\right)  _{m},
\]
and%
\begin{eqnarray*}
\left(  {x_{i}\over yt^{N-k}},q^{-1}\right)  _{m}  &  =&\prod_{j=0}^{m-1}\left(
1-{x_{i}\over yt^{N-k}q^{j}}\right) \\
&  =&\left(  -1\right)  ^{m}q^{-m\left(  m-1\right)  /2}t^{-m\left(
N-k\right)  }y^{-m}\prod_{j=0}^{m-1}\left(  x_{i}-yt^{N-k}q^{j}\right)  ,
\end{eqnarray*}
for $1\leq i\leq k$.
\end{proof}

\subsection{Homogeneity of $MS_{m^k}$ at $1-q^{m-1}t^{N-k+1}=0$\label{homogen}}\ \\

Let us describe first the solution set $Z$ of the conditions 
$$q^{m-1}t^{N-k+1}=1\mbox{ and }
q^{\left(  m-1\right)  /a}t^{\left(  N-k+1\right)  /a}\neq1$$ for any common
divisor $a$ of $m-1$ and $N-k+1$ with $a>1$.\\
 Let $d:=\gcd\left(m-1,N-k+1\right)  ,m_{0}=\frac{m-1}{d},n_{0}=\frac{N-k+1}{d}$. \\
Recall the
Euler $\phi$-function, $\phi\left(  d\right)  =\#K_{d}$ where $$K_{d}:=\left\{
k:1\leq k<d,\gcd\left(  k,d\right)  =1\right\}  .$$
 We claim that $Z$ has
$\phi\left(  d\right)  $ disjoint connected components as a subset of
$\mathbb{C}^{2}$. Let $z=re^{2\pi i\theta}\in\mathbb{C}$ with $r>0$ and
$\theta\in\mathbb{R}$, and set%
\begin{eqnarray*}
q &  =&z^{n_{0}},t=\omega z^{-m_{0}},\\
\omega &  =&e^{2i\pi\psi},
\end{eqnarray*}
then $q^{m_{0}}t^{n_{0}}=\omega^{n_{0}}$ and it is required that
$\omega^{n_{0}d}=1$ and $\left(  \omega^{n_{0}}\right)  ^{j}\neq1$ for $1\leq
j<d$. (if $d=1$ the latter restriction becomes void). For determining
connected components it suffices to let $r=1$ and project onto the $2$-torus
realized as the unit square $\left\{  \left(  \alpha,\beta\right)
:0\leq\alpha,\beta\leq1\right\}  $ with identifications $\left(
\alpha+1,\beta\right)  \equiv\left(  \alpha,\beta\right)  $ and $\left(
\alpha,\beta+1\right)  \equiv\left(  \alpha,\beta\right)  $, with the mapping
$\left(  \alpha,\beta\right)  \longmapsto\left(  e^{2\pi i\alpha},e^{2\pi
i\beta}\right)  $. The pre-image of $\left(  q,t\right)  $ is
\[
\left\{  \left(  n_{0}\theta,-m_{0}\theta+\psi\right)  :\theta\in
\mathbb{R}\right\}  .
\]
We require $\psi n_{0}d\in\mathbb{Z}$ and $\psi n_{0}j\notin\mathbb{Z}$ for
$1\leq j<d$, that is $\psi=\frac{k}{n_{0}d}$ for some $k$ and $\frac{kj}%
{d}\notin\mathbb{Z}$ (when $d=1$ the value $k=0$ is permissible). This is
equivalent to $\gcd\left(  k,d\right)  =1$. By the identification assume
$0\leq\psi<1$. So there are $\#\widetilde{K}$ solutions for $\psi$ where
$\widetilde{K}:=\left\{  k:1\leq k\leq n_{0}d,\gcd\left(  k,d\right)
=1\right\}  $, and $\#\widetilde{K}=n_{0}\phi\left(  d\right)  $. For
$k\in\widetilde{K}$ let $Z_{k}=\left\{  \left(  n_{0}\theta,-m_{0}\theta
+\frac{k}{n_{0}d}\right)  :0\leq\theta<\frac{1}{n_{0}}\right\}  $. Then
$Z_{k}$ meets $Z_{k^{\prime}}$ where $k^{\prime}=\left(  k-m_{0}d\right)
\operatorname{\ mod\ }\left(  n_{0}d\right)  $ (the first coordinate is periodic
for $\theta\mapsto\theta+$ $\frac{1}{n_{0}}$; and where $0\leq
a\operatorname{\ mod\ }b<b$ for $a\in\mathbb{Z},b\geq1$). This is continued to
show that $Z_{k}$ is connected to $Z_{k^{\prime}}$ with $k^{\prime}=\left(
k-m_{0}dj\right)  \operatorname{\ mod\ }\left(  n_{0}d\right)  ,1\leq j<n_{0}$.
Consider the set
\[
\widetilde{K}^{\prime}:=\left\{  \left(  k-m_{0}dj\right)  \operatorname{\ mod\ }%
\left(  n_{0}d\right)  ,0\leq j<n_{0},k\in K_{d}\right\}  .
\]
We claim this set coincides with $\widetilde{K}$. If $k,k^{\prime}\in K_{d}$
and $k\neq k^{\prime}$ then $\left(  k-m_{0}dj\right)  \neq\left(  k^{\prime
}-m_{0}dj^{\prime}\right)  \operatorname{\ mod\ }\left(  n_{0}d\right)  $ for any
$j,j^{\prime}$ (else $k=k^{\prime}\operatorname{\ mod\ }d$); suppose $0\leq
j<j^{\prime}<n_{0}$ then $\left(  k-m_{0}dj\right)  \neq\left(  k-m_{0}%
dj^{\prime}\right)  \operatorname{\ mod\ }\left(  n_{0}d\right)  $ or else
$m_{0}d\left(  j-j^{\prime}\right)  =n_{0}dl$ for some $l$, which is
impossible for $1\leq j-j^{\prime}<n_{0}$. Thus $\widetilde{K}^{\prime}%
\subset\widetilde{K}$ and $\#\widetilde{K}^{\prime}=\#\widetilde{K}$ and hence
$\widetilde{K}^{\prime}=\widetilde{K}$. Thus $Z$ is the disjoint union of the
connected sets $\bigcup\left\{  Z_{k+jd}:0\leq j<n_{0}\right\}  \times\left\{
r:r>0\right\}  $ for $k\in K_{d}$.\\
\begin{example}\rm
\begin{enumerate}
\item First examine the solution of $q^at^b=1$ where $\gcd(a,b)=1$ on the example $a=4$ and $b=7$.
Topologically  we can collapse the problem to the torus $[q=\exp(2i\pi \alpha),t=\exp(2i\pi \beta)]$ , $0\leq  \alpha,\beta \leq 1$. On the unit square this becomes $7\alpha+4\beta=0\ \mod\ 1$, or $\beta=-\frac74\alpha$, interpreted periodically, period 1 see fig.\ref{sgp47}.
\begin{figure}[h]
\begin{center}
\includegraphics[angle=-90]{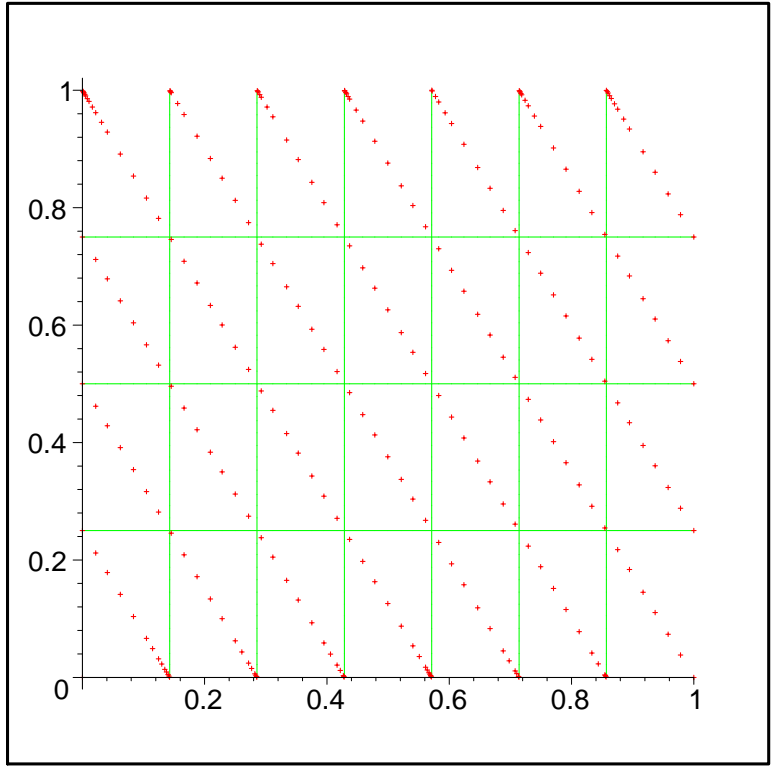}
\end{center}
\caption{Solution of $q^at^b=1$ for $a=4$ and $b=7$.\label{sgp47}}
\end{figure}
\item When $\gcd(a,b)>1$, for instance $a=8$ and $b=12$, there are several connected solutions. In the example given in fig.\ref{sgp128}, there are two solutions $2\alpha+3\beta=\frac14$ and $2\alpha+3\beta=\frac34$.
\begin{figure}[h]
\begin{center}
\includegraphics[angle=-90]{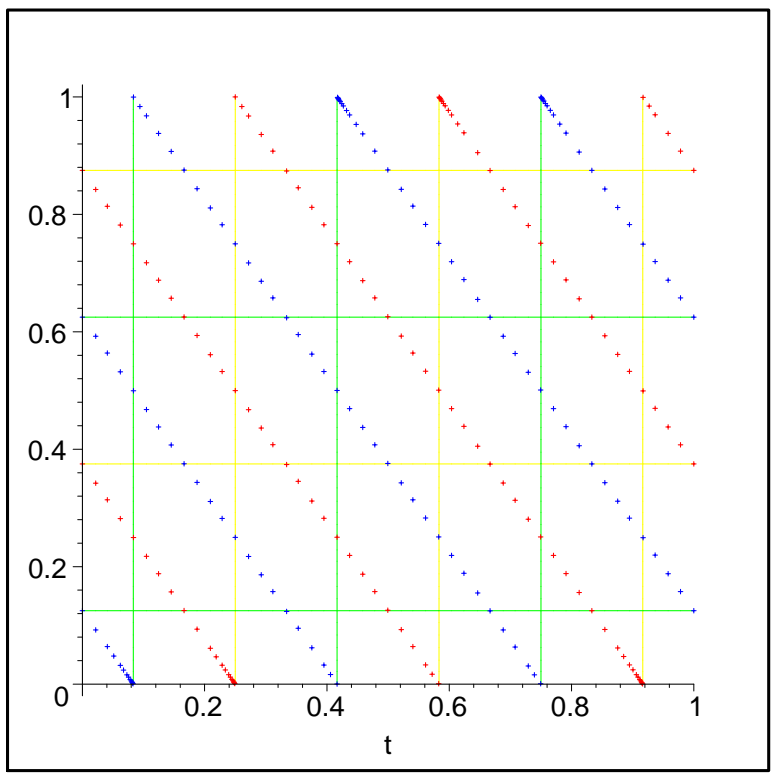}
\end{center}
\caption{Solution of $q^at^b=1$ for $a=12$ and $b=8$.\label{sgp128}}
\end{figure}
\end{enumerate}
\end{example}

\begin{proposition}\label{MS=P}
Consider $q=z^a$ and $t=\omega z^b$ a specialization of the parameters such that $1-q^{m-1}t^{N-k+1}=0$ and $1-q^{m-1\over d}t^{N-k+1\over d}\neq 0$ if $d>1$ divides $m-1$ and $N-k+1$. Then,
\[
MS_{m^k}=P_{m^k}
\]
\end{proposition}
\begin{proof}
From Lemma \ref{MS2P}, one has
\[{MS}_{m^k}=P_{m^k}+\sum_{\mu\subsetneq m^k}{\tau_\mu\over \tau_{m^k}}\left(m^k\atop\mu\right)_{q^{-1},t^{-1}}{P_{m^k}(\langle 0\rangle)\over {P}_\mu(\langle 0\rangle)}{P}_\mu.\]
From Lemma \ref{DenomP0}, ${P_{m^k}(\langle 0\rangle)\over {P}_\mu(\langle 0\rangle)}=0$ for $\mu\subsetneq m^k$. Furthermore, from Lemmas \ref{DenPrect} and \ref{denqtbin} there are no poles at $q=z^a$ and $t=\omega z^b$ in 
$P_\mu$ and $\left(m^k\atop \mu\right)_{q^{-1},t^{-1}}$ (otherwise the denominators of the two polynomials have $(1-q^{m-1}t^{N-k+1})$ as a factor which contradicts the lemmas). The result follows.\end{proof}

\section{Proof of a conjecture of Forrester}
\subsection{Clustering properties of $P_{m^k}$}
In the study of fractional quantum Hall states, Bernevig and Haldane \cite{BH} identified certain clustering conditions on rectangular Jack polynomials. These can be interpreted as factorization properties under certain specialization of the variables. In this context, Baratta and Forrester conjectured \cite{BF} a more general identity involving a rectangular Macdonald polynomial: 
\begin{equation}\label{BF}
P_{m^k}(y,yq^{\frac1\alpha},\dots,yq^{N-k-1\over\alpha},x_{N-k},\dots,x_N)=\mathcal D^{y}_{N,k,m}(x_1,\dots,x_k),
\end{equation}
where $N\geq 2k$, $\alpha=-\frac{N-k+1}{m-1}$ and $\gcd(N-k+1,m-1)=1$.\\
Propositions \ref{MS=P} and \ref{MS2D} allow us to write a more general formula.
\begin{theorem}\label{GenBF}
Let $y\in\mathbb N$. Consider $q=z^a$ and $t=\omega z^b$ a specialization of the parameters such that $1-q^{m-1}t^{N-k+1}=0$ and $1-q^{m-1\over d}t^{N-k+1\over d}\neq 0$ if $d>1$ divides $m-1$ and $N-k+1$. Then,
\[
P_{m^k}(x_1,\dots,x_k,yt^{N-k-1},\dots,yt,y)={\mathcal D}_{N,k,m}^y(x_1,\dots,x_k).
\]
\end{theorem}
\begin{proof}
Propositions \ref{MS=P} and \ref{MS2D} shows the formula for $y\neq 0$. Since the coefficients of $x^v$ in the left hand side and right hand side are polynomial in $y$, the result remains true for $y=0$.
\end{proof}

\begin{example}\rm
Let us illustrate this result with few examples:
\begin{enumerate}
\item Set $q=t^{-3}$ then
\[
P_{[2,2,0,0]}(x_1,x_2,yt,y)=(x_1-\frac yt)(x_2-\frac yt)(x_1-yt^2)(x_2-yt^2)
\]
\item Let $t=-q^{-1}$, one has:
\[
P_{[3,0]}(x_1,y)=(x_1+\frac yq)(x_1+y)(x_1+qy).
\]
whilst the result does not hold for $t=q^{-1}$:
\[
P_{[3,0]}(x_1,y)=\left(x_1^2-\left(1+\frac1q\right)^2yx_1+y^2\right)(x_1+y).
\]
\item Let $\omega=\exp\left\{2\pi i\over3\right\}$ and $t=\omega q^{-1}$. We have
\[\begin{array}{rcl}
P_{[4,4,0,0]}(x_1,x_2,ty,y)&=&
(x_1-\omega^2yq^{-2})(x_2-\omega^2yq^{-2})(x_1-\omega^2yq^{-1})(x_2-\omega^2yq^{-1})
\times\\&&\times(x_1-\omega^2y)(x_2-\omega^2y)(x_1-\omega^2yq)(x_2-\omega^2yq),
\end{array}\]
whilst the result does not hold for $t=q^{-1}$, since there is a pole in $P_{[4,4,0,0]}$:
\[
\mathrm{Numer}(P_{[4,4,0,0}(x_1,x_2,ty,y))|_{t=q^{-1}}=3{y^4\over q^6}(1-q^2)(1-q^3)^2(1-q^4)x_1^2x_2^2,
\] where $\mathrm{Numer}$ denotes the numerator of the reduced fraction.
\end{enumerate}
\end{example}
\subsection{The Jack polynomial $P^{\left(N-k+1\over 1-m\right)}_{m^k}$}
The Jack polynomial $P^{(\alpha)}_\lambda$ is usually recovered from the Macdonald polynomial $P_{\lambda}(x_1,\dots,x_N;q,t)$ by setting $q=t^\alpha$ and taking the limit $t\rightarrow 1$.\\
So as a consequence of Theorem \ref{GenBF} we obtain
\begin{corollary}
Let $N,\ k\geq 1$ and $m\geq 2$  verifying $N\geq 2k$ and $\gcd(N-k+1,m-1)=1$. Then we have
\[
P_{m^k}^{\left(N-k+1\over 1-m\right)}\left(x_1,\dots,x_k,\overbrace{y,\dots,y}^{\times N-k}\right)=\prod_{i=1}^k(x_i-y)^m.
\]
\end{corollary}

See Kakei \emph{et al.} \cite{Kakei} for related results on shifted Jack polynomials.
\subsection{Clustering properties of nonsymmetric Macdonald polynomials}
As before we use $\left(  \ast\right)  $ to indicate a term $q^{a}t^{b}$ with
$a,b\in\mathbb{Z}$ when its value can be ignored in the context of analyzing
zeros and poles. We restate a result of Lascoux \emph{et al.}\cite[Corollary 4.3 from Sahi's binomial
formula]{LRW}. For $u\in\mathbb{N}_{0}^{N}$ let $B_{u}:=\left\{  v\in\mathbb{N}%
_{0}^{N}:%
\left[u\atop v\right]\neq0,v\neq u\right\}  $. As stated before $v\in B_{u}$ implies $v^{+}\subset
u^{+}$ and $\left\vert v\right\vert <\left\vert u\right\vert $, but there are
important additional consequences as we will show.

\begin{theorem}
For $u\in\mathbb{N}_{0}^{N}$%
\[
M_{u}\left(  x\right)  =E_{u}\left(  x\right)  +\sum_{v\in B_{u}}\left(
\ast\right)
\left[u\atop v\right]_{q^{-1},t^{-1}}\frac{E_{u}\left(  \left\langle 0\right\rangle \right)
}{E_{v}\left(  \left\langle 0\right\rangle \right)  }E_{v}\left(  x\right)  .
\]

\end{theorem}

We need to use this formula for $u=\left[  m^{k},0^{N-k}\right]  $ with
$2k\leq N$ and $\left(  q,t\right)  $ in a neighborhood $\Omega$ of the set of
solutions of $q^{m}t^{N-k+1}=1$ and $q^{a}t^{b}\neq1$ for $1\leq a\leq m$ and
$0\leq b\leq N-k+1$, except for $a=m$ and $b=N-k+1$. (We note here that if
$q,t\in\mathbb{C}$ and neither $q$ nor $t$ are roots of unity then
$q^{m}t^{N-k+1}=1$ and $q^{a}t^{b}=1$ imply that $a=\frac{n}{d}m,b=\frac{n}%
{d}\left(  N-k+1\right)  $ where $d=\gcd\left(  m,N-k+1\right)  $ and
$n\in\mathbb{Z}$; roughly the equations imply $q^{am}t^{bm}=1=q^{am}%
t^{a\left(  N-k+1\right)  }$, etc.) From Knop \cite[Prop. 5.3]{Knop2} the
coefficients of $h_{q,t}\left(  v,qt\right)  M_{v}\left(  x\right)  $ are in
$\mathbb{Z}\left[  q,q^{-1},t,t^{-1}\right]  $ for any $v\in\mathbb{N}_{0}%
^{N}$ (the same multiplier works for $E_{v}\left(  x\right)  $). Furthermore
$M_{v}\left(  \left\langle v\right\rangle \right)  =\left(  \ast\right)
h_{q,t}\left(  v,q\right)  $ (from Lascoux \emph{et al.}\cite[p.8]{LRW}), and $E_{v}\left(
\left\langle 0\right\rangle \right)  =\left(  \ast\right)  \frac{\left(
t^{N}q,q,t\right)  _{v^{+}}}{h_{q,t}\left(  v,qt\right)  }$ by Proposition
\ref{hookM}. Thus we consider $h_{q,t}\left(  v,qt\right)  $ and
$h_{q,t}\left(  v,q\right)  $ for $v\in B_{u}$. We showed that $v^{+}\subset
u$ and $\left\vert v\right\vert <\left\vert u\right\vert $ implies $\left(
t^{N}q,q,t\right)  _{v^{+}}\neq0$ for $\left(  q,t\right)  \in\Omega$.

\begin{proposition}
Suppose $u\in\mathbb{N}_{0}^{N}$ and $\ell\left(  u\right)  \leq k<N$ then
$v\in B_{u}$ implies $\ell\left(  v\right)  \leq k$.
\end{proposition}

\begin{proof}
We use the \textquotedblleft extra-vanishing\textquotedblright\ theorem of
Knop \cite[Thm. 4.5]{Knop2} applied to $v$ with $M_{v}\left(  \left\langle
u\right\rangle \right)  \neq0$. This implies that there is a permutation
$w\in\mathfrak{S}_{N}$ such that for any $i$ with $1\leq i\leq N$ the
inequality $i\geq w\left(  i\right)  $ implies $v\left[  i\right]  \leq
u\left[  w\left(  i\right)  \right]  $ otherwise $v\left[  i\right]  <u\left[
w\left(  i\right)  \right]  $. By hypothesis $u\left[  j\right]  =0$ for any
$j>k$  hence $w^{-1}\left(  j\right)  \geq j$ (since $v\left[  w^{-1}\left(
j\right)  \right]  <0$ is impossible). Thus $w^{-1}\left(  j\right)  \geq j$
for each $j>k$ and by downward induction $w\left(  j\right)  =j$ for $k<j\leq
N$. By applying Knop's condition it follows that $v\left[  j\right]  \leq
u\left[  w\left(  j\right)  \right]  =u\left[  j\right]  =0$ for $j>k$.
\end{proof}

\begin{corollary}
If $u=\left[  m^{k},0^{N-k}\right]  $ and $v\in B_{u}$ then $h_{q,t}\left(
v,qt\right)  \neq0$ and $h_{q,t}\left(  v,q\right)  \neq0$ for $\left(
q,t\right)  \in\Omega$.
\end{corollary}

\begin{proof}
The factors in the hook products are of the form $1-q^{a}t^{b}$ with $1\leq
a\leq m$ and $0\leq b\leq k<N-k+1$ (the leg-lengths are bounded by $k-1$).
\end{proof}

This shows that $%
\left[m^k\atop v\right]_{q^{-1},t^{-1}}$ has no poles for $v\in B_{m^{k}}$ and $\left(  q,t\right)
\in\Omega$. The last step is to show $E_{m^{k}}\left(  \left\langle
0\right\rangle \right)  =0$ for $q^{m}t^{N-k+1}=1$ and $\left(  q,t\right)
\in\Omega$, and indeed
\[
\left(  t^{N}q;q,t\right)  _{m^{k}}=\prod_{i=1}^{k}\prod_{j=1}^{m}\left(
1-q^{j}t^{N-i+1}\right)  ;
\]
the factor with $i=k,j=m$ is $\left(  1-q^{m}t^{N-k+1}\right)  $.

\begin{proposition}\label{factE}
Suppose $u=\left[  m^{k},0^{N-k}\right],  $ $2k\leq N$, $q^{m}t^{N-k+1}=1$
and $q^{m/a}t^{\left(  N-k+1\right)  /a}\neq1$ for $a>1$ being a common
divisor of $m$ and $N-k+1$, then
\begin{eqnarray*}
M_{u}\left(  x\right)    & =&E_{u}\left(  x\right)  ,\\
E_{u}\left(  x_{1},\ldots,x_{k},yt^{N-k+1}\ldots,yt,y\right)    &
=&\mathcal{D}_{N,k,m}^{y}\left(  x_{1},\ldots,x_{k}\right)  .
\end{eqnarray*}

\end{proposition}

\begin{proof}
The proof is essentially the same as that of Proposition \ref{GenBF}.
\end{proof}

We find that the polynomial $E_{u}\left(  x\right)  $ is singular for certain
values of $\left(  q,t\right)  $ exactly when it coincides with $M_{u}\left(
x\right)  $, and in that case there is a factorization result of clustering type.

\subsection{Highest weight and singular Macdonald polynomials}
In the case of special parameter values when $MS_{\lambda}=P_{\lambda}$ by
Proposition \ref{XixiD} it follows that $MS_{\lambda}\sum_{i=1}^{N}\Xi
_{i}=P_{\lambda}\sum_{i=1}^{N}\Xi_{i}=P_{\lambda}\sum_{i=1}^{N}\xi_{i}$ and
thus $P_{\lambda}\sum_{i=1}^{N}D_{i}=0$.
For rectangular partitions, this matches with a result of one of the authors (J.-G.L.) with Th. Jolic\oe ur \cite{JL} which involves the kernel of the operator $$L_+:=\sum_{i=1}^N\prod_{j=1\atop j\neq i}^N{tx_i-x_j\over x_i-x_j}{\partial\over\partial_qx_i}$$ with
\[
{\partial\over\partial_qx_i}f(x_1,\dots,x_N)={f(x_1,\dots,x_N)-f(x_1,\dots,x_{i-1},qx_i,x_{i+1},\dots x_N)\over x_i-qx_i}.
\]
Note from a result of Baker and Forrester \cite[eq. (5.9) p12]{BaFo}, the operator $L_+$ satisfies
\[
(1-q)L_+=\sum_{i=1}^ND_i.
\]

\begin{theorem} (Jolic\oe ur-Luque)\cite{JL}\\
If $N\geq 2k$ then the Macdonald polynomial $P_{m^k}(x_1,\dots,x_N;q,t)$ belongs to the kernel of $L_+$ for the specialization $(q,t)=\left(z^{k-1-N\over d},\omega z^{m-1\over d}\right)$ where $d=\gcd(m-1,N-k-1)$ and $\omega=\exp\left\{2i\pi(1+dn)\over m-1\right\}$ with $n\in\mathbb Z$. 
\end{theorem}
We observe similar phenomena for nonsymmetric Macdonald polynomials:
Note $E_{u},\ M_{u}$ have the same eigenvalues under the action of $\xi_{i}%
,\Xi_{i}$ respectively. If for certain parameters $\left(  q,t\right)  $
$M_{u}=E_{u}$ then $E_{u}D_{i}=0$ for $1\leq i\leq N$ (since $M_{u}\Xi
_{i}=E_{u}\Xi_{i}=E_{u}\xi_{i}+E_{u}D_{i}$).
   
According to Proposition \ref{factE}: Suppose $u=\left[  m^{k},0^{N-k}\right]  $ $2k\leq N$ and $q^{m}t^{N-k+1}=1$
and $q^{m/a}t^{\left(  N-k+1\right)  /a}\neq1$ for $a>1$ being a common
divisor of $m$ and $N-k+1$, then $E_{m^k0^{N-k}}D_i=0$ for each $1\leq i\leq N$.

\section{Conclusion and perspectives}
 One of the authors (J.-G.L.) with Thierry Jolicoeur\cite{JL} investigated other families of symmetric Macdonald polynomials belonging to the kernel of $L_+$ for certain specializations of the parameters. These polynomials are indexed by staircase partitions and numerical evidence shows that they have nice factorization properties. A staircase partition $\lambda=[((\beta+1)s+r)^k,(\beta s+r)^\ell,\dots,(s+r)^\ell,0^{{\ell+1\over s-1}}]$ is defined by five integer parameters $\beta, s, r, k$ and $\ell$ such that $k\leq \ell$ and $N={\ell +1\over s-1}r+\ell(\beta+1)+k\in\N$. The corresponding specialization is $(t,q)=(z^{s-1\over g},z^{-{l+1\over g}}\omega)$ where $g=\gcd(\ell +1,s-1)$ and $\omega=\exp\left\{2i\pi(1+dg)\over s-1\right\}$. The simplest polynomials $(r=k=0)$ can be factorized for any variables $x_1,\dots,x_N$ as a $(q,t)$-discriminant. When $g=1$, this is a consequence of  a theorem proved by one of the authors (J.-G.L) with A. Boussicault \cite[Theorem 3.2]{BL}. For instance,
\[
P_{420}(x_1,x_2,x_3;q=t^{-2},t)=(-)_t(x_1-tx_2)(x_1-tx_3)(x_2-tx_1)(x_2-tx_3)(x_3-tx_1)(x_3-tx_1),
\]
where $(-)_t$ denotes a factor depending only on $t$.
But it seems that a factorization  holds also for the other cases (for special values of the variables and the parameters). Let us give a few examples:\\ \\
$\begin{array}{rcl}
P_{630}(x_1,x_2,x_3;q=-t^{-1},t)&=&(-)_t(x_1+x_2)(x_1+x_3)(x_2+x_3)(x_1-tx_2)(x_1-tx_3)\\&&(x_2-tx_1)(x_2-tx_3)(x_3-tx_1)(x_3-tx_1),\end{array}
$\\ \\
$\begin{array}{r}
P_{53000}(x_1,x_2,y,yt,yt^2;q=t^{-2},t)=(-)_t(x_1-yt^3)(x_1-yt)(x_1-yt^{-1})\\(x_2-yt^3)(x_2-yt)(x_2-yt^{-1}),
\end{array}
$\\ \\
$\begin{array}{rcl}
P_{42200}(x_1,y_1,y_1t,y_2,y_2t;q=t^{-3},t)&=&(-)_t(x_1-y_1t^2)(x_1-y_2t^2)(x_1-y_1t)(x_1-y_2t)\\&&(y_1-ty_2)(y_1-t^2y_2)(y_2-ty_1)(y_2-t^2y_1),\end{array}
$\\ \\
$
\begin{array}{r}
P_{6400000}(x_1,y_1,y_1t,y_1t^2,y_2,y_2t,y_2t^2;q=t^{-2},t)=(-)_t(x_1-y_1t^3)(x_1-y_1t^3)(x_1-y_1t^{-1})\\
(x_1-y_2t^3)(x_1-y_2t^3)(x_1-y_2t^{-1})(y_1-ty_2)(y_1-t^3y_2)(y_2-ty_1)(y_2-t^3y_1),
\end{array}
$\\ \\
$
\begin{array}{r}
P_{750^{7}}(x_1,y_1,\dots,y_1t^4,y_2,\dots,y_2t^2;q=t^2,t)=(-)_t(y_1-x_1t^3)(y_1-y_2t^3)
(y_1-x_1t^5)\\(y_1-y_2t^5)(y_2-x_1t^3)(y_2-y_1t^3)P_{420}(x_1,y_1,y_2;q=t^{-2},t),\dots
\end{array}
$\\ \\
A correct formula (and of course a proof) remains to be found.\\
Note there are also analogous formulas for singular nonsymmetric Macdonald polynomials indexed by  staircase partitions.
Consider the following examples:\\ \\
$E_{210}(x_1,x_2,x_3;q=z^{-2},t=z)=(-)_t(tx_2-x_1)(tx_3-x_1)(tx_3-x_2)$\\ \\
$\begin{array}{r}
E_{630}(x_1,x_2,x_3;q=z^{-2},t=z^3)=(-)_z
\left({\it x_2}\,z-{\it x_3} \right)  \left( -z{\it x_3}+{\it 
x_2} \right)  \left( {\it x_2}-{z}^{3}{\it x_3} \right)  \left( {\it x_1}
\,z-{\it x_3} \right)\\
 \left( -z{\it x_3}+{\it x_1} \right)  \left( {\it 
x_1}-{z}^{3}{\it x_3} \right)  \left( {\it x_1}\,z-{\it x_2} \right) 
 \left( {\it x_1}-{\it x_2}\,z \right)  \left( {\it x_1}-{\it x_2}\,{z}^{3
}\right),
\end{array}
$\\ \\
$
E_{420}(x_1,x_2,x_3;q=-t,t)=(-)_t(x_2 + x_3) (-t x_3 + x_2) (x_3 + x_1) (-t x_3 + x_1) (x_1 + x_2) (x_1 - x_2 t),
$\\ \\
$\begin{array}{r}
E_{221100}(x_1,x_2,y_1,ty_1,y_2,ty_2;q=t^{-3},t)=(-)_t(y_1 - y_2 t^2 ) (y_1 - y_2 t) (x_2 - y_2 t^2 )\\ (x_2 - t^2  y_1) (x_1 - y_2 t^2 ) (x_1 - t^2  y_1),\end{array}
$\\ \\
$
\begin{array}{r}
E_{442200}(x_1,y_1,y_1z^{2},y_2,y_2z^{2};q=z^{-3},t=z^{2})=
(-)_z(y_1-y_2z^4)(y_1-y_2z)(y_1z-y_2)\\(y_1-y_2z^2)
(x_1-y_2z^4)(x_1-y_2z)(x_1-y_1z^4)(x_1-y_1z).
\end{array}
$\\ \\
There are also equations involving permutations of $(m^k,0^{N-k})$. For instance:\\ \\	
$
E_{0022}(t,1,x_3,x_4;q=t^{-3},t)=(-)_t(tx_4-y)(tx_3-y)x_4x_3.
$\\ \\
In appendix \ref{m^k00} we show a factorization formula for permutations of $m^k0^{N-k}$.\\
The connections with singular properties remain to be investigated. In particular, if $E_{u}$ is singular then it has the same spectral vector for
$\left\{  \Xi_{i}\right\}  $ as $M_{u}$. If the spectral vector is
nondegenerate (multiplicity $=1$) then $M_{u}=E_{u}$.

\appendix
\section{About notations \label{Not}}
For parameters $a,b,c$, we consider the operator $G_{a,b,c}^i$ acting on  polynomials (or Laurent polynomials) by%
\[
G_{a,b,c}^{i}=a+\frac{bx_{i}-cx_{i+1}}{x_{i}-x_{i+1}}\left(  1-s_{i}%
\right)  ,1\leq i<N.
\]

The braid relations $G^{i}G^{i+1}G^{i}=G^{i+1}G^{i}G^{i+1}$ hold for two
families of solutions: $\left(  a,b,c\right)  =\left(  t_{1},-t_{1}%
,t_{2}\right)  $ and $\left(  a,b,c\right)  =\left(  t_{1},t_{2}%
,-t_{1}\right)  $. In both cases the quadratic relation
\[
\left(  G_{a,b,c}^{i}-t_{1}\right)  \left(  G_{a,b,c}^{i}-t_{2}\right)  =0
\]
holds. The inverses are%
\begin{eqnarray*}
\left(  G_{t_{1},-t_{1},t_{2}}^{i}\right)  ^{-1}  & =G_{1/t_{1},1/t_{2}%
,-1/t_{1}}^{i},\\
\left(  G_{t_{1},t_{2},-t_{1}}^{i}\right)  ^{-1}  & =G_{1/t_{1},-1/t_{1}%
,1/t_{2},}^{i}.
\end{eqnarray*}
Some pertinent evaluations are $1G_{a,b,c}^{i}=a$, $x_{i}G_{t_{1},-t_{1}%
,t_{2}}^{i}=-t_{2}x_{i+1}$, $x_{i+1}G_{t_{1},t_{2},-t_{1}}^{i}=-t_{2}x_{i}$.

In our applications we require the quadratic relation $\left(  G_{a,b,c}%
^{i}-t\right)  \left(  G_{a,b,c}^{i}+1\right)  =0$, and the evaluation
$1G_{a,b,c}^{i}=t$, because $1$ is (trivially) symmetric, thus $a=t,b=-1$. The
version of $T_{i}$ used in one of our previous papers \cite{DL1} is $G_{t,-t,-1}^{i}$, but Lascoux \emph{et al.}\cite{LRW} uses
$G_{t,-1,-t}^{i}$ which equals $\left(  G_{1/t,-1/t,-1}^{i}\right)  ^{-1}$
(that is, $T_{i}\left(  1/t\right)  ^{-1}$ from our notations\cite{DL1}).

With parameters $\frac{1}{t},\frac{1}{q}$ in relation to our notations \cite{DL1} the following
hold for the Lascoux \emph{et al.} \cite{LRW} versions (and $T_{i}=G_{t,-1,-t}^{i}$, $T_{i}%
^{-1}=G_{1/t,-1/t,-1}$ ):

\begin{enumerate}
\item $\xi_{i}=t^{1-i}T_{i-1}\ldots T_{1}\tau T_{N-1}^{-1}\ldots T_{i}^{-1}$
(Cherednik operator)

\item $p\tau\left(  x\right)  =p\left(  x_{N}/q,x_{1},\ldots,x_{N-1}\right)
;$

\item $D_{N}=\left(  1-\xi_{N}\right)  x_{N}^{-1}$; a simple check for the
coefficient of $\xi_{N}$: apply the operator to the constant polynomial $1$,
since $1T_{i}=t$ we have $1\xi_{i}=t^{i-N}$ and it is necessary that
$1D_{N}=0$.

\item $D_{i}=tT_{i}^{-1}D_{i+1}T_{i}^{-1}$

\item $\Xi_{i}=t^{1-i}T_{i-1}\ldots T_{1}\tau\left(  1-\frac{1}{x_{N}}\right)
T_{N-1}^{-1}\ldots T_{i}^{-1}+\frac{1}{x_{i}}$; similarly check the constants
by verifying $1\Xi_{i}=t^{i-N}$; note $\left(  \frac{1}{x_{i+1}}\right)
T_{i}^{-1}=\frac{1}{x_{i}}$.
\end{enumerate}

\section{Alternative proof of Lemma \ref{DenomP0} \label{withDum}}
Let $\lambda\subset [m^k,0^{N-k}]$ with $2k\leq N$. We have $\lambda=[m^{k'},\lambda_{k'+1},\dots,\lambda_{k''},0^{N-k''}]$ with $k'\leq k''\leq k$  and $\lambda_{k'+1}<m$.\\
The vector $\langle \lambda^-\rangle^\infty$  with $\lambda\subsetneq m^k$ is
$$
[t^{N-k''-1},\dots,t,1,*,\dots,*,q^mt^{N-1},\dots,q^mt^{N-k'},qt^{N-k''-1},\dots,qt,q,\dots]
$$
So we have to enumerate the pairs $(i,j)$ with $i=1\dots k'$ and $j=1\dots N-k''$ verifying
$\langle \lambda^-\rangle^\infty_{N-k'+i}=q^mt^{N-i}$, $\langle \lambda^-\rangle^\infty_{N+j}=qt^{N-k''-j}$  and $ {\langle \lambda^-\rangle^\infty_{N-k'+i} \over \langle \lambda^-\rangle^\infty_{N+j}}=q^{m-1}t^{N-k+1}$.
Then,  we have to enumerate the pairs $(i,j)$ with $i=1\dots k'$, $j=1\dots N-k''$ such that $N-i-(N-k''-j)=N-k+1$. Equivalently, $j=N-k-k''+1+i$. Hence, when $i$ ranges $1\dots k'$, $j$ ranges $N-k-k''+2\dots N-k+k'-k''+1$. But since $k\leq N-k$ and $k''\leq k$ we have $2\leq k-k''+2\leq N-k-k''+2$. If $\lambda\subsetneq [m^k,0^{N-k}]$ then  $k'< k$ and we have $N-k+k'-k''\leq N-k''$.\\
To resume: when $\lambda\subsetneq [m^k,0^{N-k}]$, $j$ ranges $N-k-k''+2\dots N-k+k'-k'+1\subset 1\dots N-k''$. So there are $N-k+k'-k''+1-(N-k-k''+2)+1=k'$ pairs $(i,j)$ satisfying the property. It follows that the maximal power of $(1-q^{m-1}t^{N-k+1})$ in the denominator of equality (\ref{Lasc}) is $k'$. \\
But if $\lambda=[m^k,0^{N-k}]$ then $j$ ranges $N-2k+2\dots N-k$ ($k=k'=k''$); so the maximal power is $k'-1=k-1$.\\
Similarly, we study the numerator: we have to enumerate the pair $(i,j)$ with $i=1\dots k'$, $j=1\dots N-k''$ such that $N-i-(N-k''-j)+1=N-k+1$. Equivalently, $j=N-k-k''+i$. Hence, when $i$ ranges $1\dots k'$, $j$ ranges $N-k-k''+1\dots N-k+k'-k''$.
Hence, when $i$ ranges $1\dots k'$, $j$ ranges $N-k-k''+1\dots  N-k+k'-k''$. Again $N-k-k''+1\dots  N-k+k'-k''\subset 1\dots N-k''$. So there are $N-k+k'-k''-(N-k-k''+1)+1=k'$ pairs $(i,j)$ satisfying the property. As a consequence the maximal power of $(1-q^{m-1}t^{N-k+1})$ in the numerator of (\ref{Lasc}) equals $k'$.

If $\lambda\subsetneq [m^k,0^{N-k}]$, the numerator and the denominator simplify and we have no factor of the form $(1-q^{m-1}t^{N-k+1})$ in (\ref{Lasc}). On the other hand, when $\lambda=[m^k,0^{N-k}]$, after  simplification it remains a factor  $(1-q^{m-1}t^{N-k+1})$ in (\ref{Lasc}). This completes the proof.

\section{Expressions for $M_{u}$ where $u$ is a permutation of $[m^k0^{N-k}]$\label{m^k00}}

Fix $m\geq1$ and $k$ with $1\leq k\leq N$. As before we neglect scalar
multipliers $\left(  \ast\right)  $ of the form $q^{a}t^{b}$, and then adjust
the formulae using the monic properties. Define
\[
\Omega_{N,k}:=\left\{  \beta\in\mathbb{N}^{k}:1\leq\beta\left[  1\right]
<\beta\left[  2\right]  <\ldots<\beta\left[  k\right]  \leq N\right\}
\]
For $\beta\in\Omega_{N,k}$ define $u\left(  \beta\right)  \in\mathbb{N}^{N}$
by $u\left(  \beta\right)  \left[  \beta\left[  i\right]  \right]  =m$ for
$1\leq i\leq k$ otherwise $u\left(  \beta\right)  \left[  j\right]  =0$. Let
$\sigma\left(  \beta\right)  $ denote the set $\left\{  \beta\lbrack i]:1\leq
i\leq k\right\}  $. For $1\leq i\leq k$ define $\varepsilon_{i}$ by
$\varepsilon_{i}\left[  j\right]  =\delta_{ij}$.
\begin{remark}\rm\label{rembeta}
Note if $i=\beta[j]-1\not\in \sigma(\beta)$ for some $j$ then $u(\beta).s_i=u(\beta-\varepsilon_j)$. For instance, if $N=7$ and $\beta=[1,4,5,7]$ then $u(\beta)=[m,0,0,m,m,0,m]$ and
\[u(\beta)s_6=[m,0,0,m,m,m,0]=u([1,4,5,6])=u(\beta-\varepsilon_4).\]
\end{remark}
\begin{definition}
\label{defx_beta}For $1\leq i\leq N$ let $\chi_{\beta}\left(  i\right)
=N-i-\#\left\{  l:\beta\left[  l\right]  >i\right\}  $, (if $i=\beta\left[
l\right]  $ then $\chi_{\beta}\left(  i\right)  =N-i-k+l$). Also let
$x^{\left(  \beta\right)  }$ be the point given by%
\[
x_{i}^{\left(  \beta\right)  }=\left\{
\begin{array}
[c]{c}%
x_{i},i\in\sigma\left(  \beta\right) \\
t^{\chi_{\beta}\left(  i\right)  },i\notin\sigma\left(  \beta\right)
\end{array}
\right.  .
\]

\end{definition}

The meaning of $\chi_{\beta}$ is simply that $\chi_{\beta}\left(  i\right)
=N-r_{u\left(  \beta\right)  }\left[  i\right]  $ and $x_{i}^{\left(
\beta\right)  }=\left\langle u\left(  \beta\right)  \right\rangle \left[
i\right]  $ for $i\notin\sigma\left(  \beta\right)  $. The following is one of
the main results of this section. The proof consists of several lemmas.

\begin{theorem}
\label{Mux_fact}For $\beta\in\Omega_{N,k}$%
\[
M_{u\left(  \beta\right)  }\left(  x^{\left(  \beta\right)  }\right)
=\prod_{j=1}^{k}\left(  x_{\beta\left[  j\right]  }-t^{N+j-k-\beta\left[
j\right]  }\right)  \prod_{i=1}^{m-1}\left(  x_{\beta\left[  j\right]
}-t^{N-k}q^{i}\right)
\]

\end{theorem}

\begin{lemma}
\label{Msiy}Suppose $v\in\mathbb{N}^{N}$, $v\left[  i\right]  <v\left[
i+1\right]  $ for some $i<N$ and $y$ is a point such that $M_{v}\left(
y\right)  =0$, then%
\[
M_{v.s_{i}}\left(  y\right)  =\left(  \frac{y_{i}-ty_{i+1}}{y_{i}-y_{i+1}%
}\right)  M_{v}\left(  y.s_{i}\right)  .
\]

\end{lemma}

\begin{proof}
This follows from Definition \ref{defT_i} of $T_{i}$.
\end{proof}

\begin{lemma}
\label{Mxi=0}Suppose $\beta\in\Omega_{N,k}$ and $\beta\left[  j\right]
-1\notin\sigma\left(  \beta\right)  $ for some $j$. Set $i=\beta\left[
j\right]  -1$ so that $u\left(  \beta\right)  \left[  i\right]  =0,u\left(
\beta\right)  \left[  i+1\right]  =m$ and $u\left(  \beta\right)
.s_{i}=u\left(  \beta-\varepsilon_{j}\right)  $ (see Remark \ref{rembeta}). Then%
\[
M_{u\left(  \beta\right)  }\left(  x^{\left(  \beta\right)  }.s_{i}\right)
=0.
\]

\end{lemma}

\begin{proof}
Let $v\in\mathbb{N}^{N}$ satisfy $v\left[  l\right]  \geq1$ for $l\in
\sigma\left(  \beta-\varepsilon_{j}\right)  $ and $v\left[  l\right]  =0$ for
$l\notin\sigma\left(  \beta-\varepsilon_{j}\right)  $ (in particular $v\left[
i\right]  \geq1$ and $v\left[  i+1\right]  =0$). By Knop's \textquotedblleft
extra-vanishing\textquotedblright\ theorem \cite[Thm. 4.5]{Knop2} $M_{u\left(
\beta\right)  }\left(  \left\langle v\right\rangle \right)  =0$. To prove by
contradiction suppose $M_{u\left(  \beta\right)  }\left(  \left\langle
v\right\rangle \right)  \neq0$ then there is a permutation $w\in
\mathfrak{S}_{N}$ such that for any $i$ with $1\leq i\leq N$ either $i\geq
w\left(  i\right)  $ and $u\left(  \beta\right)  \left[  i\right]  \leq
v\left[  w\left(  i\right)  \right]  $ or $i<w\left(  i\right)  $ and
$u\left(  \beta\right)  \left[  i\right]  <v\left[  w\left(  i\right)
\right]  $. Let $S_{1}=\left\{  1,\ldots,N\right\}  \backslash\sigma\left(
\beta\right)  $ and $S_{2}=\left\{  1,\ldots,N\right\}  \backslash
\sigma\left(  \beta-\varepsilon_{j}\right)  $. The sets $S_{1}$ and $S_{2}$
agree with the exception $i\in S_{1}\backslash S_{2}$ and $i+1\in
S_{2}\backslash S_{1}$. By construction $w\left(  l\right)  \in S_{2}$ implies
$l\in S_{1}$ and $l\geq w\left(  l\right)  $. Thus $w$ maps $S_{1}$ onto
$S_{2}$ and by induction (using $l\geq w\left(  l\right)  $) we see that
$w\left(  l\right)  =l$ for each $l\in S_{1}$ with $l<i$. But then $w\left(
i\right)  \geq i+1$, which is impossible.

Consider the vector $\left\langle v\right\rangle $; by construction%
\[
\left\langle v\right\rangle \left[  l\right]  =\left\{
\begin{array}
[c]{c}%
q^{v\left[  l\right]  }t^{N-r_{v}\left[  l\right]  },l\in\sigma\left(
\beta-\varepsilon_{j}\right)  \\
t^{\chi_{\beta}\left(  l\right)  },l\notin\sigma\left(  \beta-\varepsilon
_{j}\right)
\end{array}
\right.  .
\]
Note $1\leq r_{v}\left[  l\right]  \leq k$ for $l\in\sigma\left(
\beta-\varepsilon_{j}\right)  $, and $\left\langle v\right\rangle =x^{\left(
\beta\right)  }.s_{i}$ where each $x_{l}$ is of the form $q^{a}t^{b}$ with
$a\geq1$ and $N-k+1\leq b\leq N$. Now $M_{u}\left(  x^{\left(  \beta\right)
}.s_{i}\right)  $ is a polynomial in $k$ variables of degree $mk$ which
vanishes at infinitely many points of this form; by the uniqueness of the
vanishing property of shifted Macdonald polynomials $M_{u}\left(  x^{\left(
\beta\right)  }.s_{i}\right)  =0$ for all $x^{\left(  \beta\right)  }$.
\end{proof}

\begin{lemma}
Suppose $\beta=\left[  N-k+1,\ldots,N\right]  $ then%
\[
x^{\left(  \beta\right)  }=\left(  t^{N-k+1},\ldots,t,1,x_{N-k+1},\ldots
,x_{N}\right)
\]
and
\[
M_{u\left(  \beta\right)  }\left(  x^{\left(  \beta\right)  }\right)
=\prod_{i=N-k+1}^{N}\left(  x_{i}-1\right)  \prod_{j=1}^{m-1}\left(
x_{i}-t^{N-k}q^{j}\right)  .
\]

\end{lemma}

\begin{proof}
Let $v=\left[  \left(  m-1\right)  ^{k},0^{N-k}\right]  $. By the results of
Section \ref{secclust}
\[
M_{v}\left(  x_{1},\ldots,x_{k},t^{N-k-1},\ldots,1\right)  =\prod_{i=1}%
^{k}\prod_{j=0}^{m-2}\left(  x_{i}-t^{N-k}q^{j}\right)  ;
\]
if $m=1$ the product equals $1$. From $v.\Phi^{k}=u\left(  \beta\right)  $ it
follows that
\[
M_{u\left(  \beta\right)  }\left(  x\right)  =\left(  \ast\right)
M_{v}\left(  x_{N-k+1}/q,\ldots,x_{N}/q,x_{1},\ldots,x_{N-k}\right)
\prod_{i=N-k+1}^{N}\left(  x_{i}-1\right)  .
\]
Set $x_{i}=t^{N-k-i}$ for $1\leq i\leq N-k$, then%
\begin{eqnarray*}
M_{u\left(  \beta\right)  }\left(  x\right)   &  =\left(  \ast\right)
M_{v}\left(  x_{N-k+1}/q,\ldots,x_{N}/q,t^{N-k-1},\ldots,1\right)
\prod_{i=N-k+1}^{N}\left(  x_{i}-1\right) \\
&  =\left(  \ast\right)  \prod_{i=N-k+1}^{N}\left(  x_{i}-1\right)
\prod_{j=0}^{m-2}\left(  x_{i}/q-t^{N-k}q^{j}\right)  .
\end{eqnarray*}
This proves the Lemma.
\end{proof}

For $\beta\in\Omega_{N,k}$ define%
\[
F_{\beta}\left(  x\right)  :=\prod_{j=1}^{k}\left(  x_{\beta\left[  j\right]
}-t^{N+j-k-\beta\left[  j\right]  }\right)  .
\]
The special cases are $\beta=\left[  N-k+1,\ldots,N\right]  $, $F_{\beta
}=\prod_{i=N-k+1}^{N}\left(  x_{i}-1\right)  $ and $\beta=\left[
1,2,\ldots,k\right]  $, $F_{\beta}=\prod_{i=1}^{k}\left(  x_{i}-t^{N-k}%
\right)  $.

\begin{proof}
(of Theorem \ref{Mux_fact}): We use induction for steps of the form $u\left(
\beta\right)  \rightarrow u\left(  \beta\right)  .s_{i}$ applied to $u\left(
\beta\right)  $ with $u\left(  \beta\right)  \left[  i\right]  =0$ and
$u\left(  \beta\right)  \left[  i+1\right]  =m$. The previous Lemma provides
the starting point. Suppose for some $j,i$ that $\beta\left[  j\right]  =i+1$
and $i\notin\sigma\left(  \beta\right)  $, then with the identification
$x_{i}=x_{i+1}$ we have $x^{\left(  \beta\right)  }.s_{i}=x^{\left(
\beta^{\prime}\right)  }$, where $\beta^{\prime}=\beta-\varepsilon_{j}$. Then%
\[
F_{\beta^{\prime}}\left(  x^{\left(  \beta^{\prime}\right)  }\right)
=\frac{x_{i}-t^{N+j-k-i}}{x_{i}-t^{N+j-k-i-1}}F_{\beta}\left(  x^{\left(
\beta\right)  }\right)  ,
\]
because only the factor involving $x_{i}$ changes. Suppose that the claimed
formula holds for $\beta$, that is,
\[
M_{u\left(  \beta\right)  }\left(  x^{\left(  \beta\right)  }\right)
=F_{\beta}\left(  x^{\left(  \beta\right)  }\right)  \prod_{j=1}^{m-1}\left(
x_{\beta\left[  i\right]  }-t^{N-k}q^{j}\right)  .
\]
By Lemma \ref{Mxi=0} $M_{u\left(  \beta\right)  }\left(  x^{\left(
\beta\right)  }.s_{i}\right)  =0$ and by Lemma \ref{Msiy}%
\[
M_{u\left(  \beta\right)  .s_{i}}\left(  x^{\left(  \beta^{\prime}\right)
}\right)  =\left(  \frac{x_{i}^{\left(  \beta^{\prime}\right)  }%
-tx_{i+1}^{\left(  \beta^{\prime}\right)  }}{x_{i}^{\left(  \beta^{\prime
}\right)  }-x_{i+1}^{\left(  \beta^{\prime}\right)  }}\right)  M_{u\left(
\beta\right)  }\left(  x^{\left(  \beta\right)  }\right)  ,
\]
because $x^{\left(  \beta^{\prime}\right)  }.s_{i}=x^{\left(  \beta\right)  }%
$. Also $x_{i+1}^{\left(  \beta^{\prime}\right)  }=x_{i}^{\left(
\beta\right)  }=t^{\chi_{\beta}\left(  i\right)  }$ where $\chi_{\beta}\left(
i\right)  =N-i-\#\left\{  l:\beta\left[  l\right]  >i\right\}  =N-i-\left(
k-j+1\right)  $. Thus%
\begin{eqnarray*}
M_{u\left(  \beta\right)  .s_{i}}\left(  x^{\left(  \beta^{\prime}\right)
}\right)   &  =\left(\displaystyle  \frac{x_{i}-t^{N+j-i-k}}{x_{i}-t^{N+j-i-k-1}}\right)
M_{u\left(  \beta\right)  }\left(  x^{\left(  \beta\right)  }\right) \\
&  =\frac{F_{\beta^{\prime}}\left(  x^{\left(  \beta^{\prime}\right)
}\right)  }{F_{\beta}\left(  x^{\left(  \beta\right)  }\right)  }M_{u\left(
\beta\right)  }\left(  x^{\left(  \beta\right)  }\right)  ,
\end{eqnarray*}
and this proves the formula for $\beta-\varepsilon_{j}$.
\end{proof}

\begin{example}
Let $\beta=\left[  2,5,6,9\right]  \in\Omega_{10,4}$, then $u\left(
\beta\right)  =\left[  0m00mm00m0\right]  $,%
\begin{eqnarray*}
x^{\left(  \beta\right)  }  &  =\left(  t^{5},x_{2},t^{4},t^{3},x_{5}%
,x_{6},t^{2},t,x_{9},1\right)  ,\\
M_{u\left(  \beta\right)  }\left(  x^{\left(  \beta\right)  }\right)   &
=\left(  x_{2}-t^{5}\right)  \left(  x_{5}-t^{3}\right)  \left(  x_{6}%
-t^{3}\right)  \left(  x_{9}-t\right) \\
&  \times\prod_{j\in\left\{  2,5,6,9\right\}  }\prod_{i=1}^{m-1}\left(
x_{j}-t^{6}q^{i}\right)  .
\end{eqnarray*}

\end{example}

The factorization result for rectangular $E_{u}$ can be adapted to this
situation. Start with $u=\left[  m^{k},0^{N-k}\right]  $ with $2k\leq N$.

For the rest of this section assume $q^{m}t^{N-k+1}=1$ and no relation
$q^{a}t^{b}=1$ with $a<m$ or $b<N-k+1$ holds (for details see Section \ref{homogen}).

Then by Proposition \ref{factE} $M_{u}\left(  x\right)  =E_{u}\left(
x\right)  $. We claim this equality can be extended to $u\left(  \beta\right)
$ provided that $u\left(  \beta\right)  $ is a reverse lattice permutation;
this means that any substring $\left[  u\left(  \beta\right)  \left[
j\right]  ,u\left(  \beta\right)  \left[  j+1\right]  ,\ldots,u\left(
\beta\right)  \left[  N\right]  \right]  $ contains at least as many $0$'s as
$m$'s \cite[p.313]{Sta}. This condition is equivalent to
\begin{equation}
\beta\left[  j\right]  \leq N-2k+2j-1,1\leq j\leq k.\label{revlat}%
\end{equation}
As above suppose $\beta\left[  j\right]  =i+1,\beta\left[  j-1\right]  <i$ and
$v=u\left(  \beta\right)  ,v.s_{i}=u\left(  \beta-\varepsilon_{j}\right)  $.
The inverse relation is%
\[
M_{v}=\frac{\left(  1-\zeta\right)  ^{2}}{\left(  t\zeta-1\right)  \left(
\zeta-t\right)  }M_{v.s_{i}}\left(  T_{i}-\zeta\frac{1-t}{1-\zeta}\right)  ,
\]
where $\zeta=q^{m}t^{k+1+i-2j}$. The same transformation takes $E_{v.s_{i}}$
to $E_{v}$. We use induction. Suppose $M_{u\left(  \beta-\varepsilon
_{j}\right)  }=E_{u\left(  \beta-\varepsilon_{j}\right)  }$ when
$q^{m}t^{N-k+1}=1$ and both polynomials are defined (that is, the coefficients
have no poles at $q^{m}t^{N-k+1}=1$). Then the same properties hold for
$M_{u\left(  \beta\right)  }$ and $E_{u\left(  \beta\right)  }$ if $\zeta\neq
t^{-1},1,t$. Set $q^{m}=t^{-\left(  N-k+1\right)  }$ then $\zeta=t^{a}$ with
$a=2k+i-2j-N$. But $i=\beta\left[  j\right]  -1$ thus by condition
(\ref{revlat}) we have%
\begin{eqnarray*}
a &  \leq2k-2j-N+\left(  N-2k+2j-1\right)  -1&  =-2.
\end{eqnarray*}
The induction starts with $\beta=\left[  1,2\ldots,k\right]  $. We have shown

\begin{proposition}
Suppose $2k\leq N$, $\beta\in\Omega_{N,k}$, $\beta$ satisfies (\ref{revlat}),
(and $q^{m}=t^{-\left(  N-k+1\right)  }$) then%
\[
E_{u\left(  \beta\right)  }\left(  x\right)  =M_{u\left(  \beta\right)
}\left(  x\right)  .
\]

\end{proposition}

For the specialization result, let $z^{\left(  \beta\right)  }$ be defined by
$z_{i}^{\left(  \beta\right)  }=x_{i}$ if $i\in\sigma\left(  \beta\right)  $,
and $z_{i}^{\left(  \beta\right)  }=yt^{\chi_{\beta}\left(  i\right)  }$ if
$i\notin\sigma\left(  \beta\right)  $, for variables $x_{i}$ and $y$.

\begin{proposition}
Suppose $2k\leq N$, $\beta\in\Omega_{N,k}$, $\beta$ satisfies (\ref{revlat}),
(and $q^{m}=t^{-\left(  N-k+1\right)  }$), then%
\[
E_{u\left(  \beta\right)  }\left(  z^{\left(  \beta\right)  }\right)
=\prod_{j=1}^{k}\left(  x_{\beta\left[  j\right]  }-t^{N+j-k-\beta\left[
j\right]  }y\right)  \prod_{i=1}^{m-1}\left(  x_{\beta\left[  j\right]
}-t^{N-k}q^{i}y\right)  .
\]

\end{proposition}

\begin{proof}
With the notation of Definition \ref{defx_beta}%
\begin{eqnarray*}
E_{u\left(  \beta\right)  }\left(  yx^{\left(  \beta\right)  }\right)   &
=&y^{mk}E_{u\left(  \beta\right)  }\left(  x^{\left(  \beta\right)  }\right)
\\&=&y^{mk}M_{u\left(  \beta\right)  }\left(  x^{\left(  \beta\right)  }\right) \\
&  =&y^{mk}\prod_{j=1}^{k}\left(  x_{\beta\left[  j\right]  }-t^{N+j-k-\beta
\left[  j\right]  }\right)  \prod_{i=1}^{m-1}\left(  x_{\beta\left[  j\right]
}-t^{N-k}q^{i}\right)  .
\end{eqnarray*}
Replace $x_{i}$ by $x_{i}/y$ to finish the proof.
\end{proof}

\begin{example}
Let $\beta=\left[  2,5,6,9\right]  \in\Omega_{10,4}$, $q^{m}t^{7}=1$ then%
\begin{eqnarray*}
z^{\left(  \beta\right)  } &  =\left(  yt^{5},x_{2},yt^{4},yt^{3},x_{5}%
,x_{6},yt^{2},yt,x_{9},y\right)  ,\\
E_{u\left(  \beta\right)  }\left(  z^{\left(  \beta\right)  }\right)   &
=\left(  x_{2}-yt^{5}\right)  \left(  x_{5}-yt^{3}\right)  \left(
x_{6}-yt^{3}\right)  \left(  x_{9}-yt\right)  \\
&  \times\prod_{j\in\left\{  2,5,6,9\right\}  }\prod_{i=1}^{m-1}\left(
x_{j}-t^{6}q^{i}y\right)
\end{eqnarray*}

\end{example}

One expects a generalization of this result for permutations of staircase partitions.
\\ \\\noindent{\bf Acknowledgments} This paper is partially supported by the ANR project PhysComb,
ANR-08-BLAN-0243-04.

\end{document}